\documentclass[english]{IEEEtran}
\usepackage{amsthm}
\usepackage{amsmath}
\usepackage{amssymb,mathtools,MnSymbol}
\usepackage{amsfonts}
\usepackage{lscape}
\usepackage{epsf}
\usepackage{graphicx}
\usepackage{verbatim}
\usepackage[T1]{fontenc}
\usepackage[latin9]{inputenc}
\usepackage{units}
\usepackage{url}
\usepackage[normalem]{ulem}

\makeatletter
\theoremstyle{plain}
\newtheorem{thm}{\protect\theoremname}
  \theoremstyle{definition}
  \newtheorem{defn}[thm]{\protect\definitionname}
  \theoremstyle{plain}
  \newtheorem{prop}[thm]{\protect\propositionname}
  \theoremstyle{plain}
  \newtheorem{lem}[thm]{\protect\lemmaname}
  \theoremstyle{plain}
  \newtheorem{cor}[thm]{\protect\corollaryname}
\usepackage{babel}
  \providecommand{\definitionname}{\indent Definition}
  \providecommand{\lemmaname}{\indent Lemma}
  \providecommand{\propositionname}{\indent Proposition}
\providecommand{\theoremname}{\indent Theorem}
\providecommand{\corollaryname}{\indent Corollary}
\makeatother

\DeclareGraphicsExtensions{.eps}

\newtheorem{EXAMPLE}[thm]{\indent Example}

\newenvironment{example}{\begin{EXAMPLE}\rm}{\rm\qed\end{EXAMPLE}}

\newcommand{\code}{{C}}

\newcommand{\cS}{{\mathcal{S}}}

\newcommand{\dist}{{\mathsf{d}}}

\newcommand{\Sn}{{\mathbb{S}_n}}
\newcommand{\Ss}{{\mathbb{S}}}

\newcommand{\blda}{{\mbox{\boldmath $a$}}}

\newcommand{\bldc}{{\mbox{\boldmath $c$}}}

\newcommand{\bldx}{{\mbox{\boldmath $x$}}}

\newcommand{\bldy}{{\mbox{\boldmath $y$}}}

\newcommand{\half}{{\textstyle\frac{1}{2}}}

\newcommand{\ignore}[1]{}

\newlength{\Algwidth}

\begin{document}
\title{Error-Correction in Flash Memories via Codes in the Ulam Metric}
\author{Farzad~Farnoud~(Hassanzadeh), Vitaly~Skachek, and~Olgica~Milenkovic,\\
\thanks{F. Farnoud and O. Milenkovic are with the Department of Electrical and Computer Engineering, University of Illinois at Urbana-Champaign, Urbana, IL 61801, USA. V. Skachek was with the Coordinated Science Laboratory, University of Illinois at Urbana-Champaign, Urbana, IL 61801, USA. He is now with the Institute of Computer Science, Faculty of Mathematics and Computer Science, University of Tartu, J. Liivi 2-216, Tartu 50409, Estonia.}
\thanks{This work was supported by the NSF STC-CSoI 2011 CCF 0939370, the NSF CCF 0809895, and the AFRLDL-EBS AFOSR Complex Networks grants. This work was presented in part at the 2012 Information Theory and Applications Workshop, San Diego, CA, and at the 2012 International Symposium on Information Theory, Boston, MA.}}
\maketitle

\newcommand{\tpdist}{{\mathsf{d}_T}} 
\newcommand{\tldist}{{\mathsf{d}_\circ}}
\newcommand{\hdist}{{\mathsf{d}_H}}
\newcommand{\kdist}{{\mathsf{d}_{\tau}}}
\global\long\def\drcs#1{\vec{\mathsf{d}}_{\circ}\left(#1\right)} 
\newcommand{\rcs}{R_{t}} 
\global\long\def\mod#1{\left(\mbox{mod }#1\right)} 
\newcommand{\mm}[1]{P_{#1}} 
\global\long\def\what#1#2#3{#1_{\left(  #2| #3\right)}} 
\newcommand{\seg}[3]{R_{#1} \left(#2,P_{#3}\right)}
\newcommand{\set}[1]{{\left\{#1\right\}}}
\newcommand{\cs}{{\phi}} 
\newcommand{\csv}{{\varphi}} 
\newcommand{\prodd}{\bigsqcapdot}
\newcommand{\Lp}[1]{{\left(#1\right)}}
\newcommand{\sd}[2]{{D\Lp{#1,#2}}}

\global\long\def\Dld#1{\rho\left(#1\right)}

\global\long\def\slcd#1{\underline{l}\left(#1\right)}

\global\long\def\sgcd#1{\overline{l}\left(#1\right)}

\global\long\def\OSS#1{L\left(#1\right)}


\begin{abstract}
We consider rank modulation codes for flash memories that allow for handling arbitrary charge-drop errors. Unlike classical rank modulation codes used for correcting errors that manifest themselves as swaps of two adjacently ranked elements, the proposed \emph{translocation rank codes} account for more general forms of errors that arise in storage systems. Translocations represent a natural extension of the notion of adjacent transpositions and as such may be analyzed using related concepts in combinatorics and rank modulation coding. Our results include derivation of the asymptotic capacity of translocation rank codes, construction techniques for asymptotically good codes, as well as simple decoding methods for one class of constructed codes. As part of our exposition, we also highlight the close connections between the new code family and permutations with short common subsequences, deletion and insertion error-correcting codes for permutations, and permutation codes in the Hamming distance.
\end{abstract}


\section{Introduction}

Permutation codes and permutation arrays are collections of suitably chosen codewords from the symmetric group, used in applications as varied as single user communication over Gaussian channels~\cite{1445610,32151}, reduction of impulsive noise over power-lines~\cite{Blake19791,1302307}, and coding for storage~\cite{5205972}. Many instances of permutation-based codes were studied in the coding theory literature, with special emphasis on permutation arrays under the Hamming distance and rank modulation codes under the Kendall $\tau$ distance~\cite{kendall1970rankcorrelation,diaconis1977spearman,chadwick69}, \cite[Chapter 6B]{diaconis1988group}. The distances used for code construction in storage devices have mostly focused around two types of combinatorial measures, counting functions of adjacent transpositions and measures obtained via embeddings into the Hamming space~\cite{5205972, Blake19791}. This is due to the fact that such distance measures capture the displacement of symbols in retrieved messages that arise in modern nonvolatile storage systems.


One of the most prominent emerging applications of permutation codes in storage is \emph{rank modulation}. Rank modulation is an encoding scheme for flash memories that may improve the lifespan, storage efficiency and reliability of future generations of these storage devices~\cite{5452201,5205972,schartz2010correctinglimitedmagnitude,5700265}. The idea behind the modulation scheme is that information should be stored in the form of rankings of the cells' charges, rather than in terms of the absolute values of the charges. This simple conceptual coding framework may eliminate the problem of cell block erasures as well as potential cell over-injection issues~\cite{5452201,bruck2010partial-rank-modulation}. In their original formulation, rank-modulation codes represent a family of codes capable of handling errors of the form of \emph{adjacent transpositions}. Such transposition errors represent the most likely errors in a system where the cells are expected to have nearly-uniform leakage rates. But leakage rates depend on the charge of the cells, the position of the cells and on a number of external factors, the influence of which may not be adequately captured by adjacent transposition errors. For example, if a cell for a variety of reasons has a higher leakage rate than other cells, given sufficient time, the charge of this cell may drop below the charge of a large number of other cells. Furthermore, if the number of possible charge levels is large\footnote{There are two important motivations for increasing the number of charge levels. First, larger number of charge levels may enable storing more data, and second, when there are a large number of charge levels available, encoding methods such as push-to-the-top~\cite{engad2011constantWeightGrayCodes} can be used to decrease the number of times that the memory needs to be erased.}, and thus the difference between charge levels is small, a moderate charge drop may 
result in a significant drop in the cell's rank. One may argue that these processes may be modeled as a sequence of adjacent transposition errors. However, as this type of error is the result of a single error event, for the purpose of error correction it should be modeled as a \emph{single error}. This is reminiscent of the scenario where one models a sequence of individual symbol errors as a single burst error~\cite{lin-costello}.

In what follows, we present a novel approach to rank modulation coding which allows for correcting a more varied class of errors when compared to classical schemes. The focal point of the study is the notion of a translocation, a concept that generalizes an adjacent transposition in a permutation. Roughly speaking, a translocation\footnote{Note that our definition of the term translocation differs from the definition commonly used in biology. See, e.g., \cite{zhu2006translocation}.} moves the ranking of one particular element in the permutation below the rankings of a certain number of closest-ranked elements. As such, translocations are suitable for modeling errors that arise in flash memory systems, where high leakage levels for subsets of cells are expected or possible. Examples of such error events include errors due to radiation and breakdown of tunneling oxide, the latter being a prominent event in conventional poly-Si floating gate memories~\cite{CCL1,CCL2}.

A translocation may be viewed as an extension of  an adjacent transposition. In addition, translocations correspond to pairs of deletions and insertions of elements in the permutation. As a consequence, the study of translocations is closely related to the longest common subsequence problem and permutation coding under the Levenshtein and Hamming metrics~\cite{chu,Klove,levenshtein_perfect}.

Rank modulation is by now well understood from the perspective of code construction. The capacity of rank modulation codes was derived in~\cite{5205972,mazumdar2011examples,5485013}, while some practical code constructions were proposed in~\cite{5205972,5452201}, and further generalized in~\cite{6034261},~\cite{engad2011constantWeightGrayCodes} and~\cite{mazumdar2011examples}. Here, we complement the described work in terms of deriving upper and lower bounds on the capacity of translocation rank codes, and in terms of presenting constructive, asymptotically good coding schemes. Our constructions are based on a novel application of permutation interleaving, and are of independent interest in combinatorics and algebra. For the use of specialized forms of permutation interleaving in other areas of coding theory, the interested reader is referred to~\cite{klove2010PermutationArrays,schartz2010correctinglimitedmagnitude}. Furthermore, we propose decoding algorithms for translocation codes based on decoders for codes in the Hamming metric~\cite{4401563,6033769}. Finally, we also highlight the close relationships between permutation codes under a number of metrics.

The paper is organized as follows. In Section~\ref{sec:basic-defs} we provide the motivation for studying translocations as well as basic definitions used in our analysis. The properties of permutations under translocations are studied in the same section while bounds on the size of the codes are presented in Section~\ref{sec:bounds}. Code constructions are presented in Sections~\ref{sec:one-error} and Section~\ref{sec:t-errors}, while concluding remarks are given in Section~\ref{sec:conclusion}.


\section{Basic Definitions}\label{sec:basic-defs}

Throughout the paper, we use the following notation and terminology. The symbol $[n]$ denotes the set $\{ 1, 2, \cdots, n \}$. A \emph{permutation} denotes a bijection $\sigma : [n] \rightarrow [n]$, that is, for any distinct $i, j \in [n]$, we have $\sigma(i) \neq \sigma(j)$. We let $\Sn$ stand for the set of all permutations of $[n]$, i.e., the symmetric group of order $n!$. For any $\sigma \in \Sn$, we write $\sigma = (\sigma(1), \sigma(2), \cdots, \sigma(n))$, where $\sigma(i)$ is the image of $i \in [n]$ under $\sigma$. 
The identity permutation $(1, 2, \cdots, n)$ is denoted by $e$, while $\sigma^{-1}$ stands for the inverse of the permutation $\sigma$.  The product $\sigma\pi$ of two permutations $\sigma, \pi\in\Sn$ is defined so that, for each $i\in[n]$, we have $(\sigma\pi)(i) = \sigma(\pi(i))$, i.e., permutations act on the left.

For some $\sigma \in\Sn$ and $P\subseteq[n]$, the \emph{projection $\sigma_{P}$ of $\sigma$ onto $P$} is obtained from $\sigma$ by only keeping elements of $P$ and removing all other elements. For example, for $\sigma=(5,4,3,2,1)$ and $P=\{ 2,4,5\} $, we have $\sigma_{P}=(5,4,2)$. Note that $\sigma_{P}$ has length $|P|$. Next, let $\mathbb{S}(P)$ stand for the set of all permutations of elements of $P$. The identity element of $\mathbb{S}(P)$ is $e_P$, obtained from $(1,2,\cdots,n)$ by removing elements that are not in $P$.

Permutations are denoted by Greek lowercase letters, while integers and integer vectors are denoted by Latin lower case symbols. 

A \emph{transposition} $\tau(i,j)$, for distinct $i,j\in[n]$, is a permutation obtained from the identity by swapping the positions of $i$ and $j$. Namely,
\[\tau(i,j) = (1,\cdots, i-1,j,i+1,\cdots, j-1,i,j+1,\cdots,n).\] If $|i-j|=1 $, then $\tau(i,j) $ is called an \emph{adjacent transposition}.

For distinct $i,j\in[n]$, a \emph{translocation} $\cs{(i,j)}$ is a permutation obtained from the identity by moving $i $ to the position of $j$ and shifting elements between $i $ and $j $, including $ j $, by one. If $i < j$, we have \[\cs{(i,j)}=(1, \cdots, i-1, i+1, i+2, \cdots, j, i, j+1, \cdots, n)\] and if $i > j$, we have \[ \cs{(i,j)}= (1, \cdots, j-1, i, j, j+1, \cdots, i-1, i+1, \cdots, n) \; . \] For $i< j$, the permutation $\cs{(i,j)}$ is called a \emph{right-translocation} and the permutation $\cs{(j,i)}$ is called a \emph{left-translocation}. The length of a translocation $\cs{(i,j)}$ equals $|j-i|$, that is, the number of elements between $i $ and $j $, including $j $. Note that a translocation of length $k$ can be modeled by $k$ adjacent transpositions.

If the set of elements under consideration is a subset $P$ of $[n]$, for distinct $i,j\in P$, a translocation $\phi(i,j)$ over $P$ is obtained from $e_P$ by moving $i$ to the position of $j$, and shifting elements between $i $ and $j $, including $j $, by one. Right- and left-translocations over $P$ are defined similarly.

\begin{example} 
Let $\sigma = (1,3,5,7,2,4,6,8)$. We have
\begin{align*}
\sigma\cs{(3,6)}&=( 1,3,7,2,4,5,6,8),\\
\sigma\cs{(5,2)}&=( 1,2,3,5,7,4,6,8),\\
\sigma\tau(3,6)&=(1,3,4,7,2,5,6,8).
\end{align*}
Furthermore, let $P=\{2,3,5,8\}$ and $\pi=(5,8,3,2)\in\mathbb S(P)$. The translocation $ \phi(8,2)$ over $P$ equals $(8,2,3,5)$ and we have $\pi\phi(8,2) = (2,5,8,3)$. Notice that in this case, as for the case of standard permutations, the parameters in $\phi(\cdot,\cdot)$ refer to the \emph{elements} in the corresponding identity permutation, rather than positions.
\end{example}


Observe that the inverse of the left-translocation $\cs{(i,j)}$ is the right-translocation $\cs{(j,i)}$, and vice versa.

\begin{figure*}
\begin{center}
\includegraphics[width=.55\textwidth]{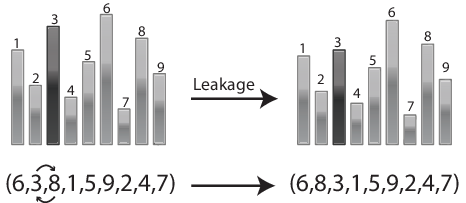}
\caption{Rank modulation codes and adjacent transposition errors.} \label{fig:adjacent-errors}
\end{center}
\end{figure*}

\begin{figure*}
\begin{center}
\includegraphics[width=.55\textwidth]{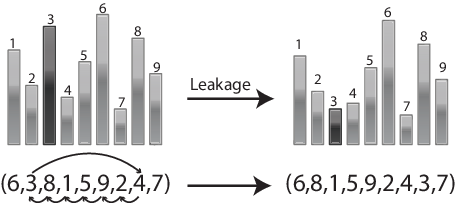}
\caption{Rank modulation codes and translocation errors caused by ``large'' drops of charge levels.} \label{fig:translocation-errors}
\end{center}
\end{figure*}

Our interest in translocations in permutations is motivated by rank modulation coding, as illustrated by the examples depicted in Figures~\ref{fig:adjacent-errors} and~\ref{fig:translocation-errors}. In classical multi-level flash memories, each cell used for storing information is subjected to errors. As a result, classical error control schemes of non-zero rate cannot be efficiently used in such systems. One solution to the problem is to encode information in terms of rankings~\cite{kendall1970rankcorrelation}, rather than absolute values of the information sequences. Consequently, data is represented by permutations and errors manifest themselves via reordering of the ranked elements. The simplest model assumes that only adjacently ranked elements may be exchanged.

This model has the drawback that it does not account for more general changes in ranks. With respect to this observation, consider the charge-drop model in Figure~\ref{fig:translocation-errors}. Here, cell number 3, ranked second, experienced a leakage rate sufficiently high to move the cell's ranking to the eighths position. This error is represented by the translocation $\phi(2,8)$. The translocation $\phi(2,8)$ corresponds to six adjacent transpositions. Nevertheless, as already argued, a translocation should be counted as a single error, and not a sequence of adjacent transposition errors.

The translocation error model may appear to be too broad to describe the phenomena arising in flash memories, as errors corresponding to translocations of small 
length arise more frequently than errors corresponding to translocations of long length. The idea of bounded length translocations (bounded burst errors) 
will be addressed in a companion paper\footnote{Note that a bounded length translocation in a permutation $\pi$ is closely related to a bounded $L_1$-metric error in $\pi^{-1}$, studied in~\cite{klove2010PermutationArrays}.}. 
We also remark that translocation errors of arbitrary length accurately model \emph{any} error that affects a single cell, and are hence suitable for modeling arbitrary charge drops of cells independently of drops of other cells, as well as read disturb and write disturb errors~\cite{Grupp2009Characterizing}. This makes them a good candidate for studying new error-control schemes in flash memories.

Next, we formalize the notion of a distance capturing translocation errors.

\begin{defn} Let $\pi, \sigma \in \Sn$. The distance $ \tldist (\pi,\sigma) $ between $\pi$ and $\sigma$ is 
defined as the minimum number of translocations needed to transform $\pi$ to $\sigma$, i.e., $\tldist(\pi,\sigma)$ equals the smallest number $m$ such that there exist a sequence of translocations $\phi_1, \phi_2, ..., \phi_m$ for which $\sigma =\pi \phi_1\phi_2 \cdots \phi_m$. 
\end{defn} 

Observe that  $\tldist\left(\cdot,\cdot\right)$ is non-negative and symmetric. It also satisfies the triangle inequality, namely, for any $\pi,\sigma$ and $\omega$ in $\Sn$, one has \[\tldist\left( \pi, \sigma \right) \leq \tldist\left( \pi, \omega \right) + \tldist\left( \omega, \sigma \right) \; . \] Therefore, it is indeed a distance metric over the space $\Sn$. 

For $\pi,\sigma\in\Sn$, the distance $\tldist(\pi,\sigma) $ is closely related to the length of the longest common subsequence of $ \pi $ and $ \sigma $, denoted by $l(\pi,\sigma)$. In fact, as shown in Prop.~\ref{prop_ulamd} below, $\tldist(\pi,\sigma) $ equals the \emph{Ulam distance}~\cite{deza1998metrics} between $\pi$ and $\sigma $, where the Ulam distance is defined as $n-l(\pi,\sigma)$. Although the Ulam distance has received some attention in the computer science community, to the best of the authors' knowledge, codes in the Ulam distance were not reported in the literature, with the notable exception of the single-error correction method by Levenshtein~\cite{levenshtein_perfect} and the asymptotically zero-rate codes presented in \cite{beame2009longest} by Beame et al.




We start our subsequent discussion with the definition of the notion of invariance.
A metric $\dist$ over $\Sn$ is right-invariant if, for all $\pi,\sigma,\omega\in\Sn$, we have $\dist(\pi,\sigma) = \dist(\pi\omega,\sigma\omega)$. Similarly, $\dist$ is left-invariant if  $\dist(\pi,\sigma) = \dist(\omega\pi,\omega\sigma)$. 
Intuitively, a right-invariant metric is invariant with respect to \emph{reordering} of elements and a left-invariant metric is invariant with respect to \emph{relabeling} of elements. 

The distance $\tldist$ is a left-invariant metric. To prove this simple observation, consider three arbitrary permutations $\pi,\sigma,\omega\in \Sn$ with $\tldist(\pi,\sigma)=m$. Then, there exists a sequence $\cs_1,\cs_2,\cdots,\cs_m$ of translocations such that $\sigma=\pi\cs_1\cs_2\cdots\cs_m$. Multiplying both sides of the previous equality by $\omega$ on the left yields $\omega\sigma=\omega\pi\cs_1\cs_2\cdots\cs_m$. This implies that $\tldist(\omega\pi,\omega\sigma)\le m=\tldist(\pi,\sigma)$. Conversely, we may repeat the same argument using $\omega\pi,\omega\sigma$ and $\omega^{-1}$ instead of $\pi,\sigma$, and $\omega$, to obtain $\tldist(\pi,\sigma)\le \tldist(\omega\pi,\omega\sigma)$. This proves the desired invariance property.

The length of the longest common subsequence of two permutations is also left-invariant. To prove this claim, let us consider again three arbitrary permutations $\pi,\sigma,\omega\in \Sn$ with $l(\pi,\sigma)=m$. Then there exists a longest common subsequence $i_1,i_2,\cdots, i_m$ of  $\pi$ and $\sigma$. Here, as anywhere else in the paper, we assume that one may choose, according to some arbitrary but fixed rule, \emph{one} longest common subsequence if the longest common sequence is not unique. It follows that $\omega(i_1),\omega(i_2),\cdots, \omega(i_m)$ is a subsequence of both $\omega\pi$ and $\omega\sigma$, and thus $l(\omega\pi,\omega\sigma)\ge l(\pi,\sigma)$. On the other hand, by considering the permutations $\omega\pi,\omega\sigma$, and $\omega^{-1}$ instead of $\pi,\sigma$, and $\omega$, it can also be shown that $l(\pi,\sigma)\ge l(\omega\pi,\omega\sigma)$. This proves that $l$ is left-invariant.

We next show that the translocation distance $\tldist\left(\pi,\sigma\right)$ equals the Ulam distance. More details about the Ulam distance and the longest common subsequence of permutations may be found in~\cite{diaconis1988group} and \cite{deza2009encyclopedia}.

\begin{prop}\label{prop_ulamd} For $\pi,\sigma\in\Sn$, the distance $\tldist\left(\pi,\sigma\right)$ equals $n-l\left(\pi,\sigma\right)$, i.e., the distance used for assessing the effect of translocations on permutation codes equals the Ulam distance between $\pi$ and $\sigma$. 
\end{prop}

\begin{proof}
By the left-invariance of $\tldist$ and $l$, we may assume that one of the permutations is the identity permutation $e$ since otherwise, instead of $\tldist\left(\pi,\sigma\right)=n-l\left(\pi,\sigma\right)$, we can show  that $\tldist\left(\sigma^{-1}\pi,e\right)=n-l\left(\sigma^{-1}\pi\right)$. It thus suffices to prove that $\tldist(\sigma,e)= n - l(\sigma)$, where $l(\sigma)=l(\sigma,e)$ is the length of the longest increasing subsequence of $\sigma $. 

Let $S_\ell$ denote the set of elements in the longest increasing subsequence of the permutation $\sigma$. Clearly, it is possible to transform $\sigma$ to $e$ with at most $n-l(\sigma)$ translocations. This can be achieved by applying translocations that each moves one element from the set $[n]\backslash S_\ell$ to its position in the identity permutation $e$. Hence, $\tldist(\sigma,e) \le n - l(\sigma)$. 

Next, we show that $\tldist(\sigma,e)\ge n - l(\sigma)$. We start with $\sigma$ and transform it to $e$ by applying a sequence of translocations. Every translocation increases the length of the longest increasing subsequence by at most one. Hence, we need at least $n-l(\sigma)$ translocations to transform $\sigma$ into $e$ and thus, $\tldist(\sigma,e) \ge n-l(\sigma)$. 
\end{proof}

Henceforth, we shall refer to $\tldist$ as the Ulam distance. The Ulam distance and other notions introduced in this section easily extend to permutations over a set $P\subseteq[n]$.

Note that a translocation may correspond to either a left- or a right-translocation. As seen from the example in Figure~\ref{fig:translocation-errors}, right-translocations correspond to general cell leakage models. On the other hand, left-translocations assume that the charge of a cell is increased above the level of other cells. We therefore also introduce the notion of the \emph{right-translocation distance}. As will be seen from our subsequent discussion, the Ulam distance is much easier to analyze than the right-translocation distance, and represents a natural lower bound for this distance. 


The Ulam distance is closely related to Levenshtein's insertion/deletion distance, defined as the number of deletions and insertions required to transform one sequence to another, and denoted by $\Dld{\cdot,\cdot}$. Levenshtein~\cite{levenshtein_perfect} showed that, for sequences of length $n$, \( \Dld{u,v}=2\left(n-l(u,v)\right).\) This equality also holds for permutations and thus \[\Dld{\pi,\sigma}=2\tldist (\pi,\sigma)\] for $\pi,\sigma\in \Sn.$ This result may be also deduced directly, by observing that a translocation consists of a deletion and an insertion. 

It is also of interest to see how the Ulam distance compares to the Kendall $\tau$ distance used in classical rank modulation coding. The \emph{Kendall $\tau$ distance} $\kdist (\pi, \sigma)$ between $\pi \in \Sn$ and $\sigma \in \Sn$ is defined as the minimum number of adjacent transpositions required to change $\pi$ into $\sigma$. A distance measure related to the Kendall $\tau$ is the transposition distance, also known as the Cayley distance. The transposition distance between two permutations $\pi$ and $\sigma$ of $\Sn$ is denoted by $\tpdist  (\pi,\sigma)$, and equals the smallest number of (not necessarily adjacent) transpositions needed to transform $\pi$ into $\sigma$. The transposition distance $\tpdist  (\pi,\sigma)$, as shown by Cayley~\cite{Cayley1849}, equals $n$ minus the number of cycles in the permutation $\sigma^{-1}\pi$.

Since a translocation of length $\ell$ can be represented as $\ell$ adjacent transpositions, and since an adjacent transposition is a translocation, it is easy to see that \[\frac{1}{n-1}\, \kdist (\pi, \sigma)  \leq \tldist (\pi,\sigma)  \leq \kdist (\pi, \sigma).\] Both the upper bound and the lower bound are tight: the upper bound is achieved for $\sigma$ obtained from $\pi$ via a single adjacent transposition, while the lower bound is achieved for, say, $\pi=e$ and $\sigma=(2,3,\ldots,n,1)$. It is also straightforward to show that the diameter of $\Sn$ with respect to the Ulam distance equals $n-1$. Observe that the above inequalities imply that the Ulam distance is not within a constant factor from the Kendall $\tau$ distance, so that code constructions and bounds specifically derived for the latter distance measure are not tight and sufficiently efficient with respect to the Ulam distance. 
 
A similar pair of bounds may be shown to hold for the Ulam distance and the Hamming distance between two permutations. The Hamming distance between  permutations $\pi$ and $\sigma$, denoted by $\hdist(\pi,\sigma)$, is defined as the number of positions $i$ for which $\pi(i)$ and $\sigma(i)$ differ.

Let $F(\pi,\sigma)=\{i\in[n]:\pi(i)=\sigma(i)\}$. The subsequence of $\pi$ consisting of elements $\pi(i),i\in F(\pi,\sigma)$, is also a subsequence of $\sigma$ and thus $\tldist\Lp{\pi,\sigma}=n-l(\pi,\sigma)\le n-|F(\pi,\sigma)|=\hdist\Lp{\pi,\sigma}$. Furthermore, since for any two permutations $\pi,\sigma\in \Sn$ one has $\hdist(\pi,\sigma)\le n$, it follows that $\hdist (\pi,\sigma)\le n \tldist(\pi,\sigma)$. Thus, \begin{equation} \frac1n \hdist(\pi,\sigma)\le \tldist\left(\pi,\sigma\right)\le \hdist\left(\pi,\sigma\right). \label{eq:transloc-hamming} \end{equation} These inequalities are sharp. For the upper-bound, consider $\pi=(1,2,\cdots,n)$ and $\sigma=(n,\cdots,2,1)$, with $n$ odd. For the lower-bound, let $\pi=(1,2,\cdots,n)$ and $\sigma=(2,3,\cdots, n, 1)$ so that $\hdist (\pi,\sigma) = n$ and $\tldist (\pi,\sigma)=1$. 

Next, we consider the transposition distance. Note that each transposition may be viewed as two translocations, implying that $\tldist(\pi,\sigma)\le 2\tpdist  (\pi,\sigma)$. It is also immediate that $\frac1{n-1}\tpdist  (\pi,\sigma)\le\tldist(\pi,\sigma)$. Hence, we have \[\frac1{n-1}\tpdist  (\pi,\sigma)\le\tldist(\pi,\sigma)\le2\tpdist  (\pi,\sigma).\]

The relationship between the Hamming distance and the transposition distance can be explained as follows. When transforming $\pi$ to $\sigma$ using transpositions, each transposition decreases the Hamming distance between the two permutations by at most two. Hence, $\tpdist  (\pi,\sigma) \ge \hdist(\pi,\sigma)/2$. Sorting a permutation of length $d$ requires at most $d$ transpositions. Thus, $\tpdist  (\pi,\sigma) \le \hdist(\pi,\sigma)$. These inequalities result in \begin{equation}\frac12 \hdist\left(\pi,\sigma\right)\le \tpdist\left(\pi,\sigma\right)\le \hdist\left(\pi,\sigma\right).\label{eq:transpos-hamming}\end{equation}
If $\pi\neq \sigma$, then $\tpdist(\pi,\sigma)\le \hdist(\pi,\sigma)-1$.

There exist many embedding methods for permutations, allowing one set of permutations with desirable properties according to a given distance to be mapped into another set of permutations with good properties in another metric space. In subsequent sections, we exhibit a method for interleaving permutations with good Hamming distance so as to obtain permutations with large minimum Ulam distance.

\subsection{Right-translocation Distance}

We describe next how to specialize the Ulam distance for the case that only right-translocations are allowed as error events.
\begin{defn}
Let $\pi, \sigma \in \Sn$ and denote by $\rcs\Lp{\pi,\sigma}$ the minimum number of right-translocations required to transform $\pi$ into $\sigma$. 
For two permutations $\pi_1, \pi_2 \in \Sn$, the right-translocation distance $\drcs{\pi_{1},\pi_{2}}$ is defined as 
\[\drcs{\pi_{1},\pi_{2}} = 2\min_{\sigma}  \max \left\{ \rcs\Lp{\pi_{1},\sigma},\rcs\Lp{\pi_{2},\sigma}\right\}  \; .\]
\end{defn}

We demonstrate next that $ \vec{\mathsf{d}}_{\circ} $ is in fact a metric by proving that it satisfies the triangle inequality; the other metric properties may be readily verified using the definition of the distance.

Consider three permutations, $\pi_{1}$, $\pi_{2}$, and $\pi_{3}$,
and let 
\begin{align*}
\sigma_{12} & =\arg\min_{\sigma}\max\left\{ \rcs\left(\pi_{1},\sigma\right),\rcs\left(\pi_{2},\sigma\right)\right\} \\
\sigma_{23} & =\arg\min_{\sigma}\max\left\{ \rcs\left(\pi_{2},\sigma\right),\rcs\left(\pi_{3},\sigma\right)\right\} .
\end{align*}
Suppose that $\alpha_{i},1\le i\le m_{\alpha}=R_{t}\left(\pi_{1},\sigma_{12}\right)$,
are right-translocations and assume that $\beta_{i},1\le i\le m_{\beta}=R_{t}\left(\pi_{2},\sigma_{12}\right)$,
are left-translocations such that 
\begin{align}
\pi_{1}\alpha_{1}\alpha_{2}\cdots\alpha_{m_{\alpha}} & =\sigma_{12},\nonumber \\
\sigma_{12}\beta_{1}\beta_{2}\cdots\beta_{m_{\beta}} & =\pi_{2}.\label{eq:p1p2}
\end{align}
Similarly, suppose that $\gamma_{i},1\le i\le m_{\gamma}=R_{t}\left(\pi_{2},\sigma_{23}\right)$,
are right-translocations and that $\delta_{i},1\le i\le m_{\delta}=R_{t}\left(\pi_{3},\sigma_{23}\right)$,
are left-translocations such that
\begin{align}
\pi_{2}\gamma_{1}\gamma_{2}\cdots\gamma_{m_{\gamma}} & =\sigma_{23},\nonumber \\
\sigma_{23}\delta_{1}\delta_{2}\cdots\delta_{m_{\delta}} & =\pi_{3}.\label{eq:p2p3}
\end{align}
Note that the existence of the sets of translocations $\left\{ \alpha_{i}\right\} ,\left\{ \beta_{i}\right\} ,\left\{ \gamma_{i}\right\} ,\left\{ \delta_{i}\right\}$ follows from the definition of $\rcs$. 

From (\ref{eq:p1p2}) and (\ref{eq:p2p3}), we have 
\begin{equation}
\pi_{1}\alpha_{1}\cdots\alpha_{m_{\alpha}}\beta_{1}\cdots\beta_{m_{\beta}}\gamma_{1}\cdots\gamma_{m_{\gamma}}\delta_{1}\cdots\delta_{m_{\delta}}=\pi_{3}.\label{eq:p1p3}
\end{equation}
Right-translocations and left-translocations have the following simple property. Suppose $\beta$ is a left-translocation and $\gamma$ is a
right-translocation. We can then find a right-translocation $\gamma'$
and a left-translocation $\beta'$ such that $\beta\gamma=\gamma'\beta'$,
where either $\gamma'$ or $\beta'$ are allowed to be the identity permutation. Hence, with slight abuse of notation, 
(\ref{eq:p1p3}) may be rewritten as 
\begin{equation}
\pi_{1}\alpha_{1}\cdots\alpha_{m_{\alpha}}\gamma'_{1}\cdots\gamma'_{m_{\gamma}}\beta'_{1}\cdots\beta'_{m_{\beta}}\delta_{1}\cdots\delta_{m_{\delta}}=\pi_{3}.\label{eq:p1p3-1}
\end{equation}
Next, let $\sigma_{13}=\pi_{1}\alpha_{1}\cdots\alpha_{m_{\alpha}}\gamma'_{1}\cdots\gamma'_{m_{\gamma}}$. Note that $\sigma_{13}$ is not required to be the minimizer of
$\max\left\{ \rcs\left(\pi_{1},\sigma\right),\rcs\left(\pi_{3},\sigma\right)\right\}$.

From (\ref{eq:p1p3-1}) and the fact that $\left\{ \alpha_{i}\right\} $
and $\left\{ \gamma'_{i}\right\} $ are right-translocations and $\left\{ \beta_{i}'\right\} $
and $\left\{ \delta_{i}\right\} $ are left-translocations, it follows
that 
\begin{align*}
\rcs\left(\pi_{1},\sigma_{13}\right) & \le m_{\alpha}+m_{\gamma}=\rcs\left(\pi_{1},\sigma_{12}\right)+\rcs\left(\pi_{2},\sigma_{23}\right),\\
\rcs\left(\pi_{3},\sigma_{13}\right) & \le m_{\beta}+m_{\delta}=\rcs\left(\pi_{2},\sigma_{12}\right)+\rcs\left(\pi_{3},\sigma_{23}\right),
\end{align*}
 and thus 
\begin{align*}
\drcs{\pi_{1},\pi_{3}} 
&\le 2\max\{\rcs(\pi_1,\sigma_{13}),\rcs(\pi_3,\sigma_{13})\}\\
&\le 2\max\big\{\rcs\left(\pi_{1}, \sigma_{12}\right) + \rcs\left(\pi_{2}, \sigma_{23}\right) ,\\
&\qquad\rcs\left(\pi_{2},\sigma_{12}\right)+\rcs\left(\pi_{3},\sigma_{23}\right)\big\}\\
&\le 2\max\{\rcs(\pi_1,\sigma_{12}),\rcs(\pi_2,\sigma_{12})\}+\\
&\qquad
2\max\{\rcs(\pi_2,\sigma_{23}),\rcs(\pi_3,\sigma_{23})\}\\
&=\drcs{\pi_1,\pi_2}+\drcs{\pi_2,\pi_3}.
\end{align*}
Hence, $  \vec{\mathsf{d}}_{\circ}$ satisfies the triangle inequality.

The definition of $\vec{\mathsf{d}}_{\circ}$ implies that for two permutations $\pi_1,\pi_2 \in \Sn$, one has $\drcs{\pi_{1},\pi_{2}} \le 2t$ if and only if there exists a permutation $\sigma \in \Sn$ such that $\rcs\Lp{\pi_1,\sigma}\le t$ and $\rcs\Lp{\pi_2,\sigma}\le t$. Hence, a code $\code$ is $t$-right-translocation correcting if and only if $\drcs{\pi_{1}, \pi_{2}} > 2t$ for all $\pi_{1},\pi_{2} \in \code$, $\pi_1 \neq \pi_2$. This means that under the given distance constraint,  it is not possible to confuse the actual codeword $\pi_1$ with another (wrong) codeword $\pi_2$. 

Observe that the following bound holds:
\[\tldist\left(\pi_{1},\pi_{2}\right)\le\drcs{\pi_{1},\pi_{2}} \; . \]
It is straightforward to characterize the minimum number of right-translocations needed to transform one permutation to another, as we show next. 

\begin{defn}
For $\pi, \sigma \in \Sn$, let
\begin{align*}
J \left(\pi,\sigma\right) &:= \big\{ i\in [n]: \exists j, \pi^{-1}(i)<\pi^{-1}(j) \\
&\qquad\mbox{ and } \sigma^{-1}(i)>\sigma^{-1}(j)\big\} \; . 
\end{align*}
\end{defn}
Note that $|J|$ is left-invariant since, for $\pi,\sigma,\omega\in \Sn$,
\begin{align*}
|J(\omega\pi,\omega\sigma)| & =\big|\big\{ i\in[n]:\exists j, \pi^{-1}\omega^{-1}(i)<\pi^{-1}\omega^{-1}(j)\\
&\qquad\mbox{ and }\sigma^{-1}\omega^{-1}(i)>\sigma^{-1}\omega^{-1}(j)\big\}\big| \\
& =\big|\big\{ i'\in[n]:\exists j',\pi^{-1}(i')<\pi^{-1}(j')\\
&\qquad\mbox{ and }\sigma^{-1}(i')>\sigma^{-1}(j')\big\}\big| \\
 & =\big|J(\pi,\sigma)\big|,
\end{align*}
where for the first equality we have used the fact that $(\omega\pi)^{-1}=\pi^{-1}\omega^{-1}$ and $(\omega\sigma)^{-1}=\sigma^{-1}\omega^{-1}$, and the second equality can be obtained by letting $i'=\omega^{-1}(i)$ and $j'=\omega^{-1}(j)$. Furthermore, using similar arguments as for the proof of left-invariance of $\tldist$, it can be shown that $\rcs$ is left-invariant.
\begin{lem}
Let $\pi,\sigma \in\Sn$. Then
\[\rcs\left(\pi,\sigma\right)=\left|J\left(\pi,\sigma\right)\right| \; . \]
\end{lem}
\begin{proof}
It suffices to show that 
\[\rcs{(\pi,e)}=\left|J(\pi,e)\right| \; , \]
where 
\[\left|J(\pi,e)\right| = \left|\left\{ i\in[n]:\exists j<i,\pi^{-1}\left(i\right)<\pi^{-1}(j)\right\} \right| \; .\]
Let $\pi_1$ be obtained from $\pi$ by applying a right-translocation that moves some element $k $ to the right. Every element of $J(\pi,e)\backslash\{k\}$ is also in $J(\pi_1,e)$ as each element of $J(\pi,e)\backslash \{k\}$ is involved in at least one inversion, which is not affected by moving $k $. Hence, $|J(\pi_1,e)| \ge |J(\pi,e)| - 1$ with equality if $ J(\pi_1,e)=J(\pi,e)\backslash \{k\} $ . Repeating the same argument yields $\rcs{(\pi,e)}\ge\left|J(\pi,e)\right|-\left|J\left(e,e\right)\right|=\left|J(\pi,e)\right|$. 

Conversely, to transform $\pi$ into $e$, it suffices to apply to each $i\in J$ the shortest right-translocation that moves this element to the smallest position $i'$, such that to the left of position $i'$ are all the elements smaller than $i$. Hence, $\rcs{(\pi,e)}\le\left|J(\pi,e)\right|.$
\end{proof}

For permutations $\pi,\sigma\in\Sn$, the difference between $\rcs\Lp{\pi,\sigma}$ and $\tldist\Lp{\pi,\sigma}$ may be as large as $n-2$. This may be seen by letting $\pi = (2,3,\dots, n,1)$ and $\sigma = e$, and observing that $\rcs\Lp{\pi,\sigma}=n-1$ and $\tldist\Lp{\pi,\sigma}=1$. Furthermore, it can be shown that this is the largest possible gap. To prove this fact, first note that $\rcs\Lp{\pi,\sigma}=0$ if and only if $\tldist\Lp{\pi,\sigma}=0$ and thus to obtain a positive gap one must have $\tldist\Lp{\pi,\sigma}\ge1$. We also have $\tldist\Lp{\pi,\sigma}\le\rcs\Lp{\pi,\sigma}\le n-1$. Hence, $1\le\tldist\Lp{\pi,\sigma}\le\rcs\Lp{\pi,\sigma}\le n-1$, which implies that the gap is at most $n-1-1=n-2$.

In the sections to follow, we mainly focus our attention on the Ulam distance.


\section{Bounds on the Size of Codes} \label{sec:bounds}
\subsection{Codes in the Ulam metric}
Henceforth, a \emph{permutation code}, or simply a code, of length $n$ and minimum distance $d$ in a metric $\dist$ refers to a subset $C$ of $\Sn$ such that for all distinct $ \pi, \sigma \in C$, we have $\dist ( \pi, \sigma)\ge d$. The term \emph{a capacity achieving code} is reserved for a code with maximum rate and a given minimum distance in a given metric space. We also let $A_\circ(n,d)$ be the maximum size of a permutation code of length $n$ and minimum Ulam distance $d$.
\begin{prop}
For all integers $n$ and $d$ with $n\ge d\ge 1$, we have \[A_\circ(n,d)\ge\frac{\left(n-d+1\right)!}{\binom{n}{d-1}}\cdot\]
\end{prop}

\begin{proof}
Let $B_\circ\Lp{r}$ be the number of permutations at Ulam distance at most $r$ from a given permutation. From left-invariance, we have $B_\circ\Lp{r}=\left|\{ \sigma:\tldist(\sigma,e)\le r\}\right| $. The permutations that are within Ulam distance $r$ from $e$ are precisely the permutations $\sigma$ with $l\left(e,\sigma\right)\ge n-r$. There are $\binom{n}{r}$ ways to choose the first $n-r$ elements of the longest common subsequence of $e$ and $\sigma$ and at most $\frac{n!}{\left(n-r\right)!}$ ways to arrange the remaining elements of $\sigma$. Hence,
\[B_\circ\Lp{r}\le\binom{n}{r}\frac{n!}{\left(n-r\right)!}\cdot\]
From the Gilbert-Varshamov bound, we have $A_\circ(n,d)\ge\frac{n!}{B_\circ\Lp{d-1}}$ and thus
\[A_\circ(n,d)\ge\frac{n!}{\binom{n}{d-1}\frac{n!}{\left(n-d+1\right)!}},\]
which completes the proof. 
\end{proof}

\begin{prop}\label{prop:tl-UB}
For all $n,d\in\mathbb Z$ with $n\ge d\ge1$, \[A_\circ(n,d)\le\left(n-d+1\right)!\;.\]
\end{prop}

\begin{proof}
We provide two proofs for this bound. The first proof is based on a projection argument first described in~\cite{5485013}, while the second proof is based on a standard counting argument.

\emph{1)} Let $C$ be a code of length $n$, size $M$, and minimum distance $d$. Let $k$ be the smallest integer such that $\pi_{\{{1,\ldots,k+1\}} } \neq \sigma_{\{{1,\ldots,k+1\}}}$ for all distinct $\pi,\sigma\in C$. Hence, $M\le\left(k+1\right)!$. By definition, there exist $\pi,\sigma\in C$ such that $\pi_{\{{1,\ldots,k\}}}=\sigma_{\{{1,\ldots,k\}}}$. So, $l\left(\pi,\sigma\right)\ge k$ and thus $d\le \tldist\left(\pi,\sigma\right)\le n-k$. Hence, $M\le\left(n-d+1\right)!$.

\emph{2)} Again, let $C$ be a code of length $n$, size $M$, and minimum distance $d$. Since the minimum distance is $d$, all $M\binom{n}{n-d+1}$ subsequences of length $n-d+1$ of the codewords of $C$ are unique. There are $\frac{n!}{\left(d-1\right)!}$ possible subsequences of length $n-d+1$. Hence, 
\[M\binom{n}{n-d+1}\le\frac{n!}{\left(d-1\right)!}\]
 which implies that $M\le\left(n-d+1\right)!$.
\end{proof}

From the two previous propositions, we obtain
\begin{equation}\label{eqn:code-card}
\frac{\left(n-d+1\right)!}{\binom{n}{d-1}}\le A_\circ(n,d)\le\left(n-d+1\right)!
\end{equation}

In the remainder of the paper, all limits are evaluated for $n\to\infty$, unless stated otherwise. Furthermore, we assume that the limits exist.

\begin{lem}\label{lem:limits}
The following results hold:
\begin{enumerate}
\item \[\lim\frac{\ln (n-d(n))!}{\ln n!}=1-\lim \frac {d(n)}n, \]
\item \[\lim\frac{\ln\frac{n!}{d(n)!}}{\ln n!}= 1-\lim \frac {d(n)}n, \]
\item \[\lim\frac{\ln \binom{n}{d(n)}}{\ln n!}=0.\]
\end{enumerate}
\end{lem}

\begin{proof} All claims follow easily from the asymptotic formula $\ln(n!) = n \ln n+O(n)$.
\end{proof}

Let $\mathcal C_\circ(d)$ denote the asymptotic capacity of translocation codes with minimum Ulam distance $d=d(n)$, that is, $\mathcal C_\circ(d)=\lim\frac{\ln A_\circ(n,d)}{\ln n!}$. 

\begin{thm}
$\mathcal C_\circ(d) = 1-\lim\frac {d(n)}n$.
\end{thm}
\begin{proof}
From \eqref{eqn:code-card}, we have
\begin{align}\label{eqn:Cn}
\frac{\ln\left(n-d+1\right)!-\ln\binom{n}{d-1}}{\ln n!}&\le \frac{\ln A_\circ(n,d)}{\ln n!}\nonumber\\
&\le\frac{\ln\left(n-d+1\right)!}{\ln n!}
\end{align}
Taking the limit of \eqref{eqn:Cn} and using Lemma~\ref{lem:limits} proves the theorem.
\end{proof}

At this point, it is worth observing that the problem of bounding the longest common subsequence in permutations has been recently studied in a combinatorial framework~\cite{beame2009longest}. There, the question of interest was to determine the minimum length of the longest common subsequence between any two distinct permutations from a set of $k$ permutations of length $n$. When translated into the terminology of translocation codes, the problem reduces to finding $d_k \left(n\right)$, the largest possible minimum Ulam distance of a set of $k$ permutations of $\Sn$. 

The bounds derived in~\cite{beame2009longest} are constructive, but they hold only in the \emph{zero-capacity domain of the code parameters}.
A more detailed description of one of the constructions of \cite{beame2009longest} is presented in the next section. The bounds of \cite{beame2009longest} imply that $d_k \left(n\right)\ge n-32\left(nk\right)^{1/3}$ for $3\le k\le\sqrt{n}$. 
Hence, for $n-32\sqrt{n}\le d\le n-32\left(3n\right)^{1/3}$, \[A_\circ(n,d)\ge\frac{1}{n}\left(\frac{n-d}{32}\right)^{3}.\] Furthermore, for $k\ge4$, $d_k \left(n\right)\ge n-\left\lceil n^{1/\left(k-1\right)}\right\rceil ^{k/2-1}$. For $k\ge2\left(1+\log_2 n\right)$, this bound is of no practical use. 

For $1+\log_2 n\le k<2\left(1+\log_2 n\right)$, one has $d_k \left(n\right)\ge n-2^{k/2-1}$ which implies that, \[A_\circ(n,d)\ge2\left(1+\log_2\left(n-d\right)\right)\] for $d\le n-\sqrt{n/2}$. 
Similar bounds can be obtained for $A_\circ(n,d)$ by assuming that $m-1<n^{\frac{1}{k-1}}\le m$ for some integer $m\le\left\lceil n^{1/3}\right\rceil $. 
Note that although these results hold for the zero-capacity regime, they still may be useful for finite codelength analysis.

\textit{Remark:} Similar bounds may be derived for the asymmetric regime of translocation error-correcting codes. 
For this purpose, let $B'(r)=\left|\left\{ \sigma:\rcs({e,\sigma})\le r\right\} \right|$ and $\vec{B}(r)=\left|\left\{ \sigma:\drcs{e,\sigma}\le r\right\} \right|$. 
Then \[\frac{n!}{B_{\circ}(2t)}\le\frac{n!}{\vec{B}(2t)}\le\vec{A}\left(n,2t+1\right)\le\frac{n!}{B'(t)},\] where $\vec{A}(n,d)$ denotes the maximum size of a permutation code with minimum right-translocation distance $d$.

\subsection{Permutation Codes in Other Metrics}

Translocation errors, and consequently, translocation error correcting codes, are difficult to analyze directly. 
On the other hand, as already pointed out, the Ulam distance is related to various other metrics well-studied in the coding theory and mathematics literature. 
Since the constructions in subsequent sections rely on codes for other distance metrics on permutations, 
we provide a brief overview of the state of the art results pertaining to the Hamming, transposition, and Kendall $\tau$ metrics. We also supplement the known findings 
with a number of new comparative results for the metrics under consideration.

\subsubsection{Hamming Metric}
Codes in the Hamming metric have a long history, dating back to the work~\cite{1445610}. The Hamming metric is a suitable distance measure for use in power line communication systems, database management and other applications.

Let $A_H(n,d)$ denote the largest number of permutations of length $n$ and minimum Hamming distance $d$. Frankl and Deza \cite[Theorem 4]{frankl-deza} and Deza~\cite{Deza1978197} showed that \[\frac{n!}{B_{H}\left(d-1\right)}\le A_H(n,d)\le\frac{n!}{\left(d-1\right)!},\] where $B_{H}\left(r\right)$ is the volume of the sphere of radius $r$ in the space of permutations with Hamming metric. Improvements of these results for some special cases were also 
obtained via linear programing methods, see for example~\cite{Tarnanen}.
 
Let $D_{i}$ denote the number of derangements of $i$ objects, i.e., the number of permutations of $[i]$ at Hamming distance $i$ from the identity permutation. It can be shown that $B_{H}\left(r\right)=1+\sum_{i=2}^{r}\binom{n}{i}D_{i}$.  Hence,
\begin{align*}
B_{H}\left(d-1\right) & = 1+\sum_{i=2}^{d-1}\binom{n}{i}D_{i}
\le\sum_{i=1}^{d-1}\frac{n!}{\left(n-i\right)!}\\
 & \le\left(d-1\right)\frac{n!}{\left(n-d+1\right)!}
\end{align*}
where the first inequality follows from the fact that $D_{i}\le i!$. Note that although a more precise asymptotic characterization for the number of derangements is known, namely \[ \lim_{\ell \rightarrow \infty} \frac{D_\ell}{\ell!} = \frac{1}{e}\;, \] the simple bound $D_{i}\le i!$ is sufficiently tight for the capacity computation.

The aforementioned results lead to \[\frac{\left(n-d+1\right)!}{d-1}\le A_H(n,d)\le\frac{n!}{(d-1)!}\cdot \] 

Let $\mathcal C_H(d)$ denote the capacity of permutation codes under the Hamming distance $d$, i.e.,
$\mathcal C_H(d)=\lim\frac{\ln A_H(n,d)}{\ln n!}$. Lemma \ref{lem:limits} implies the following theorem.

\begin{thm}\label{thm:hamming-capacity} $\mathcal{C}_H=1-\lim\frac{d(n)}{n}$.
\end{thm}
 
\vspace{.3cm} 
\subsubsection{Transposition Metric}
Let $A_{T}(n,d)$ denote the maximum size of a code with minimum transposition distance $\tpdist$ at least $d$. From \eqref{eq:transpos-hamming}, we have \[A_{H}\left(n,2d\right)\le A_{T}(n,d)\le A_{H}\left(n,d\right).\]

Using the aforementioned bounds, we have the following theorem regarding the capacity $\mathcal C_T(d)$ of permutation codes of minimum  distance $d$ in the transposition metric.
\begin{thm}\label{thm:cayley-capacity}
The capacity of permutation codes of minimum distance $d$ in the transposition metric is bounded as 
\[1-2\lim \frac{d(n)}n \le\mathcal{C}_T(d)\le1-\lim \frac{d(n)}n.\]
\end{thm}

\vspace{.3cm} 
\subsubsection{Kendall $\tau$ Metric}
Let $A_K(n,d)$ denote the largest cardinality of a permutation code of length $n$ with minimum Kendall $\tau$ distance $d$, and let $\mathcal C_K(d)=\lim \frac {\ln A_K(n,d)}{\ln n!}$. 
Barg and Mazumdar \cite[Theorem 3.1]{5485013} showed that 
\[\mathcal{C}_{K}(d)=1-\epsilon,\qquad\mbox{for } \, d=\Theta\left(n^{1+\epsilon}\right).\]
Note that for the Kendall $\tau$, the maximum distance between two permutations may be as large as $\Theta(n^2)$. On the other hand, the diameter of $\Sn$ with respect to the Ulam distance is $\Theta(n)$.

\vspace{.3cm}
\subsubsection{Levenshtein Metric}

The bounds on the size of deletion/insertion correcting codes in the more general case of codes with distinct symbols were first derived by Levenshtein in his landmark paper~\cite{levenshtein_perfect}. The lower bound relies on the use of Steiner triple systems and designs~\cite{levenshtein_perfect}. More precisely, let $D(n,q)$ be the largest cardinality of a set of $n$-subsets of the set $\{{0,1,\ldots,q-1\}}$ with the property that every $(n-1)$-subset of $\{{0,1,\ldots,q-1\}}$ is a subset of at most one of the $n$-subsets. Then the following results holds for the cardinality $\mathcal{A}_L(n,q)$ of the largest single-deletion correcting codes consisting of codewords in  $\{0,1,\ldots,q-1\}^n$ with distinct symbols~\cite{levenshtein_perfect}: \[ (n-1)! D(n,q)\, \leq  \mathcal{A}_L(n,q)  \leq \frac{q!}{n\,(q-n+1)!}\cdot \] 


\section{Single Error Correcting Codes for Translocations and Right-translocations} \label{sec:one-error}

This section contains constructions for single-translocation error detecting and single-translocation error correcting codes. For the latter case, we exhibit two constructions, one for translocations and another for right-translocations. 

\subsection{Detecting a Single Translocation Error}
We start by describing a code that can detect a single translocation error. From the discussion in Section~\ref{sec:basic-defs}, recall that the Ulam distance is half of the Levenshtein distance and thus any single-deletion correcting code may be used for detecting a single translocation error. An elegant construction for single-deletion correcting codes for permutation was described by Levenshtein in~\cite{levenshtein_perfect}. The resulting code has cardinality $\left(n-1\right)!$ and is optimal since, from Proposition \ref{prop:tl-UB}, we have
\[
A_\circ\left(n,2\right)\le\left(n-2+1\right)!=\left(n-1\right)!.
\]
Hence, $A_\circ\left(n,2\right)=\left(n-1\right)!$. 
 
Levenshtein's construction is of the following form.

Let 
\[
W_{2}^{n}=\left\{ u\in\left\{ 0,1\right\} ^{n}:(n+1)\, | \, \sum_{i=1}^{n}iu\left(i\right)\right\} 
\]
 where $a\, | \, b$ denotes that $a$ is a divisor of $b$. For $ \sigma \in\Sn$,
let $Z\left(\sigma \right)=\left(z(1),\cdots,z\left(n-1\right)\right)$
be a vector with 
\[
z\left(i\right)=\begin{cases}
0, & \qquad\mbox{if } \sigma \left(i\right)\le \sigma \left(i+1\right),\\
1, & \qquad\mbox{if } \sigma \left(i\right)> \sigma \left(i+1\right),
\end{cases}\qquad i\in\left[n-1\right].
\]
The code 
\begin{equation}
C=\left\{ \sigma \in\Sn:Z\left(\sigma \right)\in W_{2}^{n-1}\right\}\label{eq:leven-code}
\end{equation}
 of size $\left(n-1\right)!$ is capable of correcting a single deletion. Hence, this code can detect a single translocation error as well.

Let $\mathbb S_n^{(t)}$ be the set of sequences $\sigma$ of length $n-t$ that can be obtained from some permutation in $\Sn$ by $t$ deletions. In other words, $\mathbb S_n^{(t)}$ is the set of words of length $n-t$ from the alphabet $[n]$ without repetitions. A code $C_p\subseteq\Sn$ is a \emph{perfect} code capable of correcting $t$ deletions if, for every $\sigma \in\mathbb S_n^{(t)}$, there exists a unique $\pi \in C_p$ such that $\sigma$ can be obtained from $\pi$ by $t$ deletions. It was shown in~\cite{levenshtein_perfect} that $C$ in (\ref{eq:leven-code}) is a perfect code capable of correcting a single deletion.

The minimum Levenshtein distance of $C_p$ is $2(t+1)$ and thus the minimum Ulam distance of $C_p$ is $t+1$. Since the size of $\mathbb S_n^{(t)}$ equals $\binom{n}{n-t}\left(n-t\right)!$ and $\binom{n}{t}$ elements of $\mathbb S_n^{(t)}$ can be obtained by $t$ deletions from each $\sigma \in C_p$, we have that $\left|C_p\right|=\left(n-t\right)!$. Recall from Prop.~\ref{prop:tl-UB} that the size of a code with minimum Ulam distance $t+1$ is $\le (n-t)!$. Thus a perfect code capable of correcting $t$ deletions, if it exists, is a rate-optimal code in the Ulam metric. Although conditions for the existence of such codes were investigated in~\cite{levenshtein_perfect}, both necessary and sufficient conditions are known only for a small number of special cases.


In the next two subsections, we describe codes capable of correcting a single right-translocation error and codes capable of correcting a single translocation error. In the constructions, we make use of a single-transposition error detecting code, described next.

\paragraph*{A single-transposition error detecting code} For $\sigma_{1},\sigma_{2}\in\Sn $, let $\tpdist\left(\sigma_{1},\sigma_{2}\right)$ as before denote the transposition distance between $\sigma_{1}$ and $\sigma_{2}$. The parity of a permutation $\sigma$ is defined as the parity of $\tpdist(\sigma,e)$. It is well-known that applying a transposition to a permutation changes the parity of the permutation, and also that, for $n\ge2$, half of the permutations in $\Sn $ are even and half of them are odd\footnote{More precisely, the symmetric group can be partitioned into the alternating group and its coset.}. Hence, the code $C$ containing all even permutations of $\Sn $ is a single-transposition error detecting code of length $n$ and cardinality $n!/2$.

\subsection{Correcting a Single Right-translocation Error}

Next, we present a construction for codes that correct a single right-translocation error. For this purpose, we first define the operation of permutation interleaving and the operation of code interleaving.
\begin{defn}
For vectors $\sigma_i,i\in[k]$, of lengths $m_i$ with $m_1\ge m_2\ge \cdots \ge m_k \ge m_1 -1$, the \emph{interleaved vector} $\sigma=\sigma_1\circ\sigma_2\circ\cdots\circ\sigma_k$ is obtained by alternatively placing the elements of  $\sigma_1,\sigma_2,\cdots,\sigma_k$ in order. That is, 
\begin{equation}
\sigma(j) = \sigma_i(\lceil j/k\rceil), \qquad 1\le j\le \sum_{i=1}^k m_i
\end{equation}
 where $i\equiv j\mod{k}$.
For a class of $k$ codes $C_i,i\in[k]$, let
\begin{equation}
C_1\circ\cdots\circ C_k = \{\sigma_1\circ\cdots\circ\sigma_k:\sigma_i\in C_i, i\in [k]\}.
\end{equation}
For example, for vectors $\sigma$ and $\pi$ of length $m$, we have 
\begin{equation}
\sigma\circ\pi = \Lp{\sigma(1),\pi(1),\sigma(2),\pi(2),\cdots,\sigma(m),\pi(m)}
\end{equation}
and for vectors $\sigma$ and $\pi$ of lengths $m$ and $m-1$ respectively, we have
\begin{equation}
\sigma\circ\pi = \Lp{\sigma(1),\pi(1),\sigma(2),\pi(2),\cdots,\sigma(m)}.
\end{equation}

\end{defn}
The following proposition introduces codes that can correct a single right-translocation error. The decoding algorithm is contained in the proof of the proposition.
\begin{prop}
\label{prop:right-translocation} Let $P_i, i\!=\!1, 2,$ be the set of odd and even numbers in $[n]$, respectively, and let $C_i$ be the set of even permutations of $P_i$ for $i = 1,2$. The interleaved code $C=C_1\circ C_2$ corrects a single right-translocation error.\end{prop}
\begin{figure}
\begin{center}
\includegraphics[width=\columnwidth]{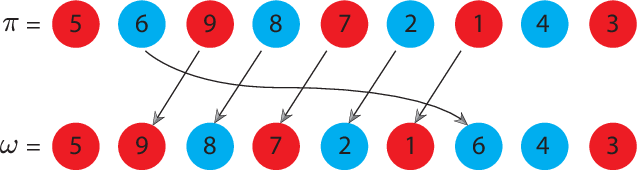}
\caption{The effect of the right-translocation error $\phi(2,7)$ on the codeword $\pi=(5,6,9,8,7,2,1,4,3)$. The result is the word $\omega = (5,9,8,7,2,1,6,4,3)$. } 
\label{fig_rtcode}
\end{center}
\end{figure}

\begin{proof}
Given the permutation $\omega\notin C$, we want to find the unique $\pi\in C$ such that $\omega=\pi \phi(i,j),\,i< j$. An example is shown in \figurename~\ref{fig_rtcode}, with $\omega = (5,9,8,7,2,1,6,4,3)$ and 
$\pi$ unknown to the decoder. 

The $k$th element of $\omega$ is out of place if $k\not\equiv\omega(k)\mod2$. It is easy to see that 
\[
i=k_{\min}:=\min\left\{ k:k\not\equiv\omega(k)\mod 2\right\} ,
\]
i.e., $i$ equals the smallest integer $k$ such that the $k$th element of $\omega$ is out of place. In the example shown in \figurename~\ref{fig_rtcode}, $i=2$. Finding $j$ is slightly more complicated since we must consider two different cases depending on the parity of the length $j-i$ of the right-translocation.

Let $k_{\max}:=\max\left\{ k:k\not\equiv\omega(k)\mod 2\right\}.$
If $j-i$ is odd, then $j=k_{\max}$. Otherwise, $j=k_{\max}+1$. That is, the right-translocation error is either $\phi(i,k_{\max})$ or $\phi(i,k_{\max}+1)$. Thus, the codeword $\pi$ either equals $\pi'=\omega\phi(k_{\max},i)$ or equals $\pi''=\omega\phi(k_{\max}+1,i)$. 
In the example of \figurename~\ref{fig_rtcode}, we have
\begin{equation*}
\begin{split}
\pi'&= (5,6,9,8,7,2,1,4,3),\\
\pi''&=(5,4,9,8,7,2,1,6,3).
\end{split}
\end{equation*}
To find which of the two cases is correct, we proceed as follows.

Since $\omega = \pi''\phi(i,k_{\max}+1)$ and $\pi'=\omega\phi(k_{\max},i)$, we have 
\begin{equation*}
\begin{split}
\pi'&= \pi''\phi(i,k_{\max}+1)\phi(k_{\max},i)\\
&=\pi''\tau(i,k_{\max}+1).
\end{split}
\end{equation*}

Recall that if $j-i$ is odd, then $j=k_{\max}$, and if $j-i$ is even, then $j=k_{\max}+1$. Hence, $i\equiv k_{\max} +1 \mod2$. In both $\pi'$ and $\pi''$,
 the parity of the elements is the same as the parity of their positions. Thus, the transposition $ \tau(i,k_{\max}+1) $ affects only elements of the same parity as $i$. Hence, if $i$ is odd, then $\pi'_{P_2}=\pi''_{P_2}= \omega_{P_2}=\pi_{P_2}$ and if $i$ is even, then $\pi'_{P_1}=\pi''_{P_1}= \omega_{P_1}=\pi_{P_1}$.

Without loss of generality, assume that $i$ is even. Then $\pi'_{P_1}=\pi''_{P_1}= \omega_{P_1}=\pi_{P_1}$ and the subwords $\pi'_{P_2}$ and $\pi''_{P_2}$ differ in one transposition. Since $C_2$ has minimum transposition distance two, only one of $\pi'_{P_2}$ and $\pi''_{P_2}$ belongs to $C_2$, and so $\pi$ can be uniquely determined as being either equal to $\pi'$ or $\pi''$.
\end{proof}

The cardinality of the interleaved code $C_1 \circ C_2$ equals $\frac{1}{4} \left\lceil\frac{n}{2}\right\rceil! \left\lfloor\frac{n}{2}\right\rfloor!$, and its rate asymptotically equals
\[
\frac{\ln\left(1/4\right)+2\ln\left\lfloor\frac{n}{2}\right\rfloor!}{\ln n!}\sim\frac{n\ln n+O\left(n\right)}{n\ln n+O\left(n\right)}\sim1.
\]

\subsection{Correcting a Single Translocation Error}

The construction of the previous subsection can be extended to generate codes capable of correcting a single translocation error as stated in the following proposition. Although the proposition is stated for $n$ being a multiple of three, it can be easily extended to other cases. 

\begin{prop}
Suppose $n$ is a multiple of three. Let $P_i, i\!=\!1, 2,3,$ be the set of numbers in $[n]$ that are equal to $i$ modulo three, and let $C_i$ be the set of even permutations of $P_i$ for $i=1,2,3$. The interleaved code $C=C_1\circ C_2\circ C_3$ corrects a single translocation error.
\end{prop}
\begin{proof}
Suppose that $\pi$ is the stored permutation, $\omega$ is the retrieved permutation, and that the error is the translocation $\phi(i,j)$. If $|i-j|=1$, then $\phi(i,j)$ can be easily identified. Suppose that $|i-j|>1$. The translocation $\phi(i,j)$ moves $|i-j|$ elements of $\pi$ one position to the left, provided that $i<j$, or one position to the right, provided that $j<i$. In either case, one element moves in the ``opposite direction'' from the other elements. Hence, for $|i-j|>1$ the direction of the translocation (left or right) can be identified. 

Once the direction of the translocation is known, $i $ can be found as follows: if the error is a right-translocation, then 
\[i = \min\{k:\omega(k)\not\equiv k \mod{3}\},\]
and if the error is a left-translocation, then
\[i = \max\{k:\omega(k)\not\equiv k \mod{3}\}.\]

For simplicity, suppose the error is a right-translocation. The proof for left-translocations is similar. Let $k_{\max}:=\max\left\{ k:k\not\equiv\omega(k)\mod 3\right\}.$ We have the following three cases. If $j-i\equiv 0\mod3 $, then $j=k_{\max}+1$ and \[\omega(k_{\max})\equiv \omega(k_{\max}+1)\mod3.\] If $j-i\equiv 1\mod3 $, then  $j=k_{\max}$ and \[\omega(k_{\max})\not\equiv \omega(k_{\max}+1)\mod3.\] Finally, if $j-i\equiv 2\mod3 $, then $j=k_{\max}$ and \[\omega(k_{\max})\equiv \omega(k_{\max}+1)\mod3.\]

So, if $\omega(k_{\max})\not\equiv \omega(k_{\max}+1)\mod3$, then $j=k_{\max}$ and $\pi$ is uniquely determined as $ \omega \phi(k_{\max},i) $. Otherwise, the error is either $\phi(i,k_{\max})$ or $\phi(i,k_{\max}+1)$. Let $\pi'=\omega \phi(k_{\max},i)$ and $\pi'=\omega \phi(k_{\max}+1,i)$. Then, $\pi'=\pi''\tau(i,k_{\max}+1)$ and similar to the proof of Prop.~\ref{prop:right-translocation}, it can be shown that $\pi'$ and $\pi''$ are not both in $C $. Hence, $\pi$ can be determined as either being equal to $\pi'$ or $\pi''$.
\end{proof}


\begin{example}
Consider the single translocation-correcting code for $n=12$. For this case, we have $P_{1}=\{1,4,7,10\},P_{2}=\{2,5,8,11\},$ and $P_{3}=\{3,6,9,12\}$.

Suppose that the stored codeword is $\pi$, the error is $\phi(i,j)$, and the retrieved word is \[\omega=(1,6,10,8,3,7,5,11,12,4,2,9).\] 

Given $\omega$, the decoder first identifies the elements that are out of order, i.e., elements that are not equivalent to their positions modulo three -- in this case, $\{6,10,8,3,7,5\}$. Since more than two elements are out of order, we have $|i-j|>1$. Furthermore, since more than two elements have moved one position to the left, $\phi$ is a right-translocation. Observe that $k_{\max} = 7$ and that $\omega(k_{\max})\equiv\omega(k_{\max}+1)\mod3$. Hence, we let 
\begin{equation*}
\begin{split}
\pi' &= \omega\phi(7,2)=(1,5,10,8,7,11,4,2)\\
\pi'' &= \omega\phi(8,2)=(1,11,10,8,7,5,4,2).
\end{split}
\end{equation*}
We then have $\pi'_{\mm2}=(5,8,11,2)$ and $\pi''_{\mm2}=(11,8,5,2)$. Since only $\pi''_{\mm2}$ is an even permutation, the error is $\phi(2,8)$ and thus $\pi=(1,11,6,10,8,3,7,5,12,4,2,9)$.
 \end{example}

The cardinality of the code equals $\left(\frac{1}{2}\left(\frac{n}{3}\right)!\right)^{3},$ while its rate equals
\[
\frac{3\ln\left(1/2\right)+3\ln\left(\frac{n}{3}\right)!}{\ln n!}\sim\frac{n\ln n+O\left(n\right)}{n\ln n+O\left(n\right)}\sim1.
\]


\section{$t$-translocation Error Correcting Codes}\label{sec:t-errors}

We describe next a number of general constructions for $t$-translocation error-correcting codes. We start with an extension of the interleaving methods from Section~\ref{sec:one-error}.

\subsection{Interleaving Codes in Hamming Metric}\label{sec:construction1}

We construct a family of codes with Ulam distance $2t+1$, length $n=s\left(2t+1\right)$ for some integer $s\ge4t+1$, and cardinality $M=\Lp{A_H(s,4t+1)}^{2t+1}$, where $A_H(s,d)$, as before, denotes the maximum size of a permutation code with length $s $ and minimum Hamming distance $d$. The construction relies on the use of $2t+1$ permutation codes, each with minimum Hamming distance at least $4t+1$. First, we present the proposed construction and then prove that the minimum Ulam distance of the code is at least $2t+1$.

For a given $n$ and $t$, where $n\equiv0\mod{2t+1}$, partition the set $[n]$ into $2t+1$ classes $\mm{i}$, each of size $s$, with \begin{equation} \label{eq:defn-Pi} \mm{i}=\left\{ j\in[n]:j\equiv i\mod{2t+1}\right\}, \qquad i\in[2t+1]. \end{equation} For example, for $t=2$ and $n=45,$ one has 
\begin{align*} \mm{1}&=\left\{ 1,6,\cdots,41\right\} ,\\\mm{2}&=\left\{ 2,7,\cdots,42\right\}, \\\mm{3}&=\left\{ 3,8,\cdots,43\right\}, 
\end{align*} and so on.

For $i\in [2t+1]$, let $C_{i}$ be a permutation code over $\mm{i}$ with minimum Hamming distance at least $4t+1$.\footnote{It is clear that instead of using permutation codes for interleaving, one can also use codes with
distinct symbols such as those described in~\cite{Davydov}.} The code $C$ is obtained by interleaving the codes $C_{i}$, i.e., $C=C_1 \circ \cdots \circ C_{2t+1}$, and is referred to as an interleaved code with $2t+1$ classes. In the interleaved code, the $s$ elements of $\mm i$ occupy positions that are equivalent to $i$ modulo $2t+1$.

The following theorem provides a lower-bound for the minimum Ulam distance of $C$. The proof of the theorem is presented after stating the required definitions, and three technical lemmas.

\begin{thm}\label{thm:min-dist}
Assume we are given three positive integers $s$, $t$, $n=s(2t+1)$, and a partition of $[n]$ of the form given in \eqref{eq:defn-Pi}. If, for $i\in[2t+1]$, $C_i$ is a permutation code over $P_i$ with minimum Hamming distance at least $4t+1$, then $C=C_1\circ\cdots\circ C_{2t+1}$ is a permutation code over $[n]$ with minimum Ulam distance greater than or equal to $2t+1$. \end{thm}
\begin{cor}
For the code $C$ of Theorem \ref{thm:min-dist} and distinct $\sigma,\pi\in C$, the length of the longest common subsequence of $\pi$ and $\sigma$ is less than $n-2t$.
\end{cor}

For convenience, we introduce an alternative notation for translocations. Let the mapping $\psi:\Sn\to\Sn$ be defined as follows. For a permutation $ \sigma \in\Sn$, an integer $\ell$, and $a\in[n]$, let $\psi {(a,\ell)} \sigma$ denote the permutation obtained from $\sigma$ by moving the element $a$ exactly $|\ell|$ positions to the right if $\ell\ge0$ and to the left if $\ell\le0$. In other words, for any permutation $\sigma\in\Sn$ and $a\in[n]$, \[\psi {(a,\ell)}\sigma=\sigma\phi\Lp{\sigma^{-1}(a),\sigma^{-1}(a)+\ell}.\] For example, we have  $\psi {(4,3)}(3,4,2,5,1)=(3,2,5,1,4)=(3,4,2,5,1)\phi(2,5)$.  Note that the mapping $\psi$ is written multiplicatively. 
Furthermore, with slight abuse of terminology, $\psi {(a,\ell)}$ may also be called a translocation.

Consider $\sigma,\pi \in \Sn$ with distance $\tldist(\sigma,\pi)=m$. A \emph{transformation} from $\sigma$ to $\pi$ is a sequence $\psi_1,\psi_2,\cdots,\psi_m$ of translocations such that $\pi = \psi_m\cdots\psi_2\psi_1 \sigma$.

Let $b_1< b_2<\cdots<b_m$ be the elements of $[n]$ that are not in the longest common subsequence of $\sigma$ and $\pi$. Each $b_k$ is called a \emph{displaced} element. The set $\set{b_1,\cdots,b_m}$ is called the \emph{set of displaced elements} and is denoted by $\sd \sigma \pi$.

The \emph{canonical} transformation from $\sigma$ to $\pi$ is a transformation $\psi_m\cdots\psi_2\psi_1$ with $\psi_k=\psi\Lp{b_k,\ell_k}$ for appropriate choices of $\ell_k,k\in [m]$. In other words, the canonical transformation operates only on displaced elements and corresponds to a shortest sequence of translocations that transform $ \sigma $ to $ \pi $.

As an example, consider $\sigma=(1,2,3,4,5,6,7,\cdots,15)$ and $\pi=(1,3,4,5,6,2,7,\cdots,15)$. Here, $n=15$, $t=1$, and $s=3$. The canonical transformation is $\psi {\left(2,4\right)}$ and we have \begin{equation} \pi=\psi {\left(2,4\right)}\sigma=(1,3,4,5,6,2,7,\cdots,15).\label{eq:running-ex} \end{equation} In this example, $\sd \sigma \pi = \set 2$.

Let $\pi_l = \psi_l\cdots\psi_2\psi_1\sigma$ for $1\le l\le m$. An element $a$ is \emph{moved over} an element $b$ in step $j$ if there exists a translocation in the canonical transformation $\psi=\psi {(a,\ell)}$ such that $a$ is on the left (right) of $b$ in $\pi_{j-1}$ and on the right (left) side of $b$ in $\pi_j$. That is, $\psi$ moves $a$ from one side of $b$ to the other side. In the above example with $\psi {(2,4)}$, $2$ is moved over $4$ but it is not moved over $7$.

An element $k\in\left[2t+1\right]$ is called a \emph{$\sigma,\pi-$pivot}, or simply a pivot, if no element of $\mm{k}$ is displaced, i.e., $\mm k \cap \sd \sigma \pi =\emptyset$. In the example corresponding to (\ref{eq:running-ex}), the pivots are $1$ and $3$.

For $I\subseteq [2t+1]$, define $P_I$ as $\cup_{i\in I}P_i$. Also, recall that for $\omega \in\Sn $ and a set $P$, $\omega_P$ denotes the projection of $\omega$ onto $P$. For example, for $t=1$, $n=15$, $\omega=(1,2,3,\cdots,15)$, and $I=\left\{ 1,3\right\} $, we have $\omega_{P_I}=(1,3,4,6,7,9,10,12,13,15)$. We say that $\omega_{P_I}$ has a \emph{correct order} if for every $i,j\in I,i<j$, elements of $\mm{i}$ and $\mm{j}$ appear alternatively in $\omega_{P_i \cup P_j}$, starting with an element of $\mm{i}$. In the example above, $\omega_{P_I}$ has a correct order.

Consider $\omega \in\Sn$ and suppose that  $\omega_{P_i} =(a_1,a_2,\cdots, a_s)$. The elements of the set $\mm i=\set{a_1,\cdots,a_s}$ may be viewed as separating subsequences of $\omega$ consisting of elements not in $\mm i$. That is, we may write \[\omega = r_0 a_1 r_1 a_2 \cdots a_s r_s\] where the $r_l$'s are non-intersecting subsequences of $[n]\backslash \mm i$. For each $l,0\le l \le s$, the subsequence $r_l$ is called the \emph{$l$-th segment of $\omega$ with respect to $\mm i$} and is denoted by $\seg l \omega i$. Such a segmentation is shown next for the permutation \[\omega=(c_3,b_4,a_1,b_2,b_3,a_2,a_3,c_1,a_4,b_1,c_2,c_4).\] Each segment is marked with a bracket:
\[\omega=(
\undergroup{c_3,b_4}
,a_{1},
\undergroup{b_2,b_3}
,a_{2},
\undergroup{\ \ }
a_{3},
\undergroup{c_1}
,a_{4},
\undergroup{b_1,c_2,c_4}).\]
To better visualize the subsequences in question, we may replace each element of $\mm i$ by $\star$ and write $\omega$ as 
\[\omega=(c_3,b_4\ \star\ b_2,b_3\ \star \ \star\  c_1 \ \star\  b_1,c_2,c_4).\]
We have, for example, $\seg 0 \omega i=(c_3,b_4)$ and $\seg 2 \omega i = ()$.

\begin{defn}
Consider $i,j\in\left[2t+1\right]$, $\mm{i}=\left\{ a_{1},\cdots,a_{s}\right\}$, and $\omega \in\Sn $. Suppose, without loss of generality, that 
$\omega_{P_i}=(a_{1},a_2,\cdots, a_{s})$. The sequence $\what \omega ji$ is defined as follows.
\begin{itemize}
\item If $i=j$, then $\what \omega ji=\what \omega ii$ equals $\omega_{P_i}$. 
\item If $j>i$, then, for $1\le l \le s$, let $\what \omega ji(l)=\seg l {\omega_{P_i \cup P_j}} i$ whenever $\seg l {\what \omega ji} i$ has length one, and let $\what \omega ji(l)=\epsilon$ otherwise. 
Here, $\epsilon$ is a special notational symbol.
\item If $j<i$, then, for $1\le l \le s$, let $\what \omega ji(l)=\seg {l-1} {\omega_{P_i \cup P_j}} i$ whenever $\seg {l-1} {\what \omega ji} i$ has length one, and let $\what \omega ji(l)=\epsilon$ otherwise.
\end{itemize}
\end{defn}

As an example, if $\mm i=\set{ a_{1},a_{2},a_{3},a_{4},a_{5}} $, $\mm j=\set{ b_{1},b_{2},b_{3},b_{4},b_{5}} $, $j<i$, and $\omega_{P_i \cup P_j} =(b_{1},a_{1},a_{2},b_{2},b_{3},a_{3},b_{4},a_{4},b_{5},a_{5})$, the segments of 
$\omega_{P_i \cup P_j}$ with respect to $P_i$ are $\Lp{b_1}$, $()$, $\Lp{b_2,b_3}$, $\Lp{b_4}$, and $\Lp{b_5}$, in the given order, and we have $\what \omega ji=(b_{1},\epsilon,\epsilon, b_{4},b_{5})$.


\begin{lem} \label{lem:I_star_sym} Consider the interleaved code $C$ of Theorem \ref{thm:min-dist} and let $\sigma \in C$. Furthermore, let $ \omega \in \Sn$ be such that $\tldist(\sigma,\omega)\le t$. There exists at least one subset $I\subseteq\left[2t+1\right]$ of size at least $t+1$ such that $\omega_{P_I}$ has a correct order. \end{lem}

\begin{proof} There are at most $t$ displaced elements, and thus, at most $t$ classes containing a displaced element. Hence, there exist at least $2t+1-t=t+1$ classes without any displaced elements and, consequently, at least $t+1$ $\sigma,\omega-$pivots. Let $I$ be the set consisting of these pivots. It is clear that $\omega_{P_I}$ obtained in this way has a correct order which proves the claimed result. \end{proof}

\begin{lem}\label{lem:projection} For all positive integers $s$ and $t$ and all permutations $\sigma,\omega\in \Sn $, with $n=(2t+1)s$, if $i^*$ is a $\sigma,\omega-$pivot, then for $j\in [2t+1]$, \[\hdist\Lp{\what \sigma j{i^*} , \what \omega j{i^*}}\le 2 \tldist\Lp{\sigma,\omega}.\] \end{lem}

\begin{proof}
Assume  $\tldist\Lp{\sigma,\omega}= m$ and let  $\psi_m\cdots\psi_2\psi_1$ be the canonical transformation from $\sigma$ to $\omega$, so that $\omega=\psi_m\cdots\psi_2\psi_1\sigma$. We prove the lemma by induction on $m$. Clearly, if $m=0$, then \[\hdist\Lp{\what \sigma j{i^*} , \what \omega j{i^*}}=0.\] 

Let $\pi=\psi_{m-1}\cdots\psi_2\psi_1 \sigma$. As the induction hypothesis, assume that \[\hdist\Lp{\what \sigma j{i^*} , \what \pi j{i^*}}\le 2(m-1).\] By the triangle inequality, it suffices to show that \begin{equation}\label{eq:induct} \hdist\Lp{\what \pi j{i^*} , \what \omega j{i^*}}\le 2. \end{equation}

Suppose $\psi_m = \psi {(b,\ell)}$ so that $\omega=\psi_m \, \pi$. Since $i^*$ is a pivot, we have $b\notin \mm{i^*}$. We consider two cases: $b\notin\mm{j}$ and $b\in\mm{j}$. First, suppose $b\notin\mm{j}$. Since $b\notin {P_{i^*}\cup P_j}$, we have $\pi_{P_{i^*}\cup P_j}=\omega_{P_{i^*}\cup P_j}$ and thus $\hdist\Lp{\what \pi j{i^*} , \what \omega j{i^*}}=0$.

On the other hand, suppose $b\in P_j$. Then, $b$ appears in $\seg k {\pi_{P_{i^*} \cup P_j}} {i^*}$ of $\pi_{P_{i^*} \cup P_j}$ and in $\seg l {\omega_{P_{i^*} \cup P_j}}  {i^*}$ of $\omega_{P_{i^*} \cup P_j}$,  for some $k,l$. The only segments affected by the translocation $\psi_m$ are $\seg k {\pi_{P_{i^*} \cup P_j}} {i^*}$ and $\seg l {\omega_{P_{i^*} \cup P_j}} {i^*}$, and thus, for $p\in [2t+1]\backslash\set{l,k}$, we have $\seg p {\pi_{P_{i^*} \cup P_j}} {i^*}=\seg p {\omega_{P_{i^*} \cup P_j}} {i^*}$. Hence, for $p\in [2t+1]\backslash\set{l,k}$, we find $\what \pi j{i^*}(p)=\what \omega j{i^*}(p)$, implying that $\hdist\Lp{\what \pi j{i^*} , \what \omega j{i^*}}\le 2$. 
\end{proof}

\begin{lem} \label{lem:what} Consider the interleaved code $C$ of Theorem \ref{thm:min-dist}. Let $\sigma \in C$ and $\omega \in \Sn$ such that $\tldist(\sigma,\omega)\le t$. If $I \subseteq[2t+1]$ is of size at least $t+1$ and $\omega_{P_I}$ has a correct order, then 
\begin{enumerate}
\item for each $i\in I$, $\hdist\left(\sigma_{P_i},\omega_{P_i}\right)\le2t$ and,
\item for $i\in I$ and $j\notin I$, $\hdist \left( \sigma_{P_j},\what \omega ji \right)\le2t$.
\end{enumerate}
\end{lem}

\begin{proof}
Since there are at most $t$ classes containing displaced elements and $I$ has size at least $t+1$, there exists a pivot $i^{*}\in I$. 
Then, by Lemma \ref{lem:projection}, 
\[\hdist\Lp{\what \sigma i {i^*}, \what \omega i {i^*}}\le 2t.\]
Since $\sigma$ is a codeword in $C$, by construction, we have $\what \sigma i {i^*}=\sigma_{P_i}$. Furthermore, since $\omega_{P_I}$ has a correct order, we have $\what \omega i {i^*} = \omega_{P_i}$. Hence, 
\[\hdist\left(\sigma_{P_i},\omega_{P_i}\right)\le2t.\]

To prove the second part, we proceed as follows. Assume $\psi_m\cdots\psi_2\psi_1$, with $m=\tldist(\sigma,\omega)$, is the canonical transformation from $\sigma$ to $\omega$ so that $\omega = \psi_m\cdots\psi_2\psi_1 \sigma$.

We first show that $\psi_m\cdots\psi_2\psi_1$ may be decomposed into four parts
\[
\omega=\left(\psi^{(j)}_{t^{(j)}}\cdots\psi^{(j)}_1
\left(\psi^{(i)}_{t^{(i)}}\cdots\psi^{(i)}_1
\left(\tau_{t_\tau}\cdots \tau_1
\left(\psi'_{t'}\cdots\psi'_1 \sigma \right)\right)\right)\right),
\] with $t'+t^{(i)}+t^{(j)}=m$ such that 
\begin{align}
\psi_k '  =\psi \left(a'_k ,\ell'_k \right),&\quad a_k '\notin \mm{i}\cup \mm{j},k\in\left[t'\right],\label{eq:decomp}\\
\tau_k   =\tau\left(a_k ,b_k \right),&\quad a_k,b_k \in \mm{i},k\in\left[t_\tau\right],\nonumber \\
\psi_k ^{(i)}  =\psi {\left(a_k ^{(i)},\ell_k ^{(i)}\right)},&\quad a_k ^{(i)}\in \mm{i},k\in\left[t^{(i)}\right],\nonumber \\
\psi_k ^{(j)}  =\psi {\left(a_k ^{(j)},\ell_k ^{(j)}\right)},&\quad a_k ^{(j)}\in \mm{j},k\in\left[t^{(j)}\right],\nonumber 
\end{align}
and such that no $\psi_k^{(i)}$ moves $a^{(i)}_k$ over an element of $\mm {i^*}$.

It can be easily verified that any two translocations $\psi {\left(a,\ell_{1}\right)}$ and 
$\psi {\left(b,\ell_{2}\right)}$ ``commute''. That is, for any permutation $\pi$, we can find translocations $\psi {\left(a,\ell_{3}\right)}$ and 
$\psi {\left(b,\ell_{4}\right)}$ such that $\psi {\left(a,\ell_{1}\right)}\psi {\left(b,\ell_{2}\right)}\pi=\psi {\left(b,\ell_{4}\right)}\psi {\left(a,\ell_{3}\right)}\pi$. 
Thus, we have the decomposition
\[
\omega=\left(\psi^{(j)}_{t^{(j)}}\cdots\psi^{(j)}_1
\left(\psi''_{t^{(i)}}\cdots\psi''_1
\left(\psi'_{t'}\cdots\psi'_1 \sigma \right)\right)\right)
\] with $t'+t^{(i)}+t^{(j)}=m$ such that 
\begin{align}
\psi_k '  =\psi \left(a'_k ,\ell'_k \right),&\quad a_k '\notin \mm{i}\cup \mm{j},k\in\left[t'\right],\nonumber\\
\psi''_k  =\psi {\left(a''_k ,\ell''_k\right)},&\quad a''_k \in \mm{i},k\in\left[t^{(i)}\right],\nonumber \\
\psi_k ^{(j)}  =\psi {\left(a_k ^{(j)},\ell_k ^{(j)}\right)},&\quad a_k ^{(j)}\in \mm{j},k\in\left[t^{(j)}\right].\nonumber 
\end{align}
Furthermore, it is easy to see that we may write $\psi''_{t^{(i)}}\cdots\psi''_1$ as $\psi^{(i)}_{t^{(i)}}\cdots\psi^{(i)}_1\tau_{t_\tau}\cdots \tau_1$ with
\begin{align}
\tau_k   =\tau\left(a_k ,b_k \right),&\quad a_k ,b_k \in \mm{i},k\in\left[t_\tau\right],\nonumber 
\end{align} 
such that no $\psi_k^{(i)}$ moves $a^{(i)}_k$ over an element of $\mm {i^*}$. Hence, for any permutation $\omega$ one can write a 
decomposition of the form~\eqref{eq:decomp}.

Let $\omega'=\tau_{t_\tau}\cdots \tau_1 \psi'_{t'}\cdots\psi'_1 \sigma$, and $\omega^{(i)}=\psi^{(i)}_{t^{(i)}}\cdots\psi^{(i)}_1\omega'$, so that 
$\omega=\psi^{(j)}_{t^{(j)}}\cdots\psi^{(j)}_1\omega^{(i)}$.
By the triangle inequality
\begin{align*}
\hdist \left( \sigma_{P_j} ,\what \omega ji\right)
&\le \hdist \left( \sigma_{P_j} ,\what{\omega'}ji\right)\\
&\ +   \hdist \left(\what{\omega'}ji,\what{\omega^{(i)}}ji\right)\\
&\ +   \hdist \left(\what{\omega^{(i)}}ji,\what \omega ji\right).
\end{align*}
It is clear that $\hdist \left(\sigma_{P_j},\what {\omega'}j i\right)=0$. 

Next, consider $\what{\omega'}ji$ and its transform $\what{\omega^{(i)}}ji$ induced by the translocations $\psi_k ^{(i)}, k\in \left[t^{(i)}\right]$. 
Note that $\omega'_{P_\set{j,i,i^*}}$ has a correct order. Since no translocation $\psi_k^{(i)}$ moves $a_k^{(i)}$ over an element of $\mm {i^*}$, each $\psi_k^{(i)}$ moves $a_k^{(i)}$ over at most one element of $\mm {j}$. 
Thus, each $\psi_k^{(i)}$ can modify at most two segments and we have $\hdist \left(\what{\omega'}ji,\what{\omega^{(i)}}ji\right)\le2t^{(i)}$. Furthermore, each $\psi_k ^{(j)}$ modifies at most two segments and thus  $\hdist \left(\what{\psi^{(i)}}ji,\what wji\right)\le2t^{(j)}.$ Hence,
\[
\hdist \left(\sigma_{P_j} ,\what wji\right)\le0+2t^{(i)}+2t^{(j)}\le2m.
\]
\end{proof}


\begin{proof} \emph{(Theorem \ref{thm:min-dist})} 
Suppose the minimum Ulam distance of $C$ is less than $2t+1$. Then, for two distinct codewords $\pi,\sigma \in C$, there exists an $\omega \in \Sn$ such that $\tldist(\pi,\omega)\le t$ and $\tldist(\sigma,\omega)\le t$.

Since $\sigma \neq \pi$, there exists $k\in[2t+1]$ such that $\pi_{P_k}\neq \sigma_{P_k}$, which implies that $\hdist\Lp{\pi_{P_k},\sigma_{P_k}}\ge 4t+1$. Since $\tldist(\sigma,\omega)\le t$, by Lemma \ref{lem:I_star_sym}, there exists $I\subseteq [2t+1]$ of size at least $t+1$ such that $\omega_{P_I}$ has a correct order.

If $k\in I$, by Lemma \ref{lem:what}, $\hdist\Lp{\sigma_{P_k},\omega_{P_k}}\le 2t$ and $\hdist\Lp{\pi_{P_k},\omega_{P_k}}\le 2t$, which together imply $\hdist\Lp{\sigma_{P_k},\pi_{P_k}}\le 4t$. 

On the other hand, if $k\notin I$, by Lemma \ref{lem:what}, for any $i\in I$, $\hdist\Lp{ \sigma_{P_k},\what wki}\le 2t$ and $\hdist\Lp{ \pi_{P_k},\what wki}\le 2t$, 
which again imply  $\hdist\Lp{\sigma_{P_k},\pi_{P_k}}\le 4t$. 

Hence, by contradiction, the minimum distance of $C$ is at least $2t+1$.
\end{proof}

The rate of the aforementioned translocation correcting codes based on interleaving may be estimated as follows. The cardinality of the interleaved code of length $n$ and minimum distance $d=d(n)$ is at least $\left(A_{H} \left(\left\lfloor\frac{n}{d}\right\rfloor,2d-1\right)\right)^{d}$ for odd $d$, and $\left(A_{H} \left(\left\lfloor \frac{n}{d+1} \right\rfloor,2d+1 \right)\right) ^{d+1}$ for even $d$. The construction is applicable only if $d(n) \le \sqrt{n/2}-1$, in which case the asymptotic rate of the interleaved code equals 
\begin{align*}
&\lim\frac{\ln|C|}{\ln n!}  =\lim\frac{d(n)\ln A_{H}\left(\frac{n}{d(n)},2d(n)\right)}{\ln n!}\\
 &\quad =\lim\frac{\ln A_{H}\left(\frac{n}{d(n)},2d(n)\right)}{\ln\left(n/d(n)\right)!}\frac{d(n)\ln\left(n/d(n)\right)!}{\ln n!}\\
 &\quad =\left(1-2\lim\frac{d^{2}(n)}{n}\right)
\lim \frac{n\ln n-n\ln d(n)+O\left(n\right)} {n\ln n+O\left(n\right)},
\end{align*}
where we used Theorem~\ref{thm:hamming-capacity} to obtain the last equality. For example, if $d(n)=n^{\beta}$, $\beta < 1/2$, then
\[1-2\lim\frac{d^{2}(n)}{n} = 1\]
and one obtains a translocation error-correcting code of rate
\begin{align*}
\lim\frac{\ln|C|}{\ln n!} & =
\lim\frac{n\ln n-\beta n\ln n+O\left(n\right)} {n\ln n+O\left(n\right)}=
1-\beta.
\end{align*}

In the next section we describe a modification of the interleaving procedure, which, when applied recursively, improves upon the code rate $1-\beta$.

\subsection{Interleaving Codes in the Hamming Metric and the Ulam Metric}

The interleaving approach described in the previous subsection may be extended in a straightforward manner.  Rather than interleaving permutation codes with good Hamming distance, as in Section \ref{sec:construction1}, one may construct a code in the Ulam metric by interleaving a code with good Ulam distance and a code with  good Hamming distance. Furthermore, this approach may be implemented in a recursive manner. In what follows, we explain one such approach and show how it leads to improved code rates as compared to simple interleaving.

We find the following results useful for our recursive construction method.

\begin{lem}
Let $\sigma, \pi \in \Sn$ be two permutations, such that $\tldist(\sigma, \pi) = 1$. 
Then, there exist at most three positions $i$, $i \in [n-1]$, such that for some $j = j(i) \in [n-1]$:
\begin{itemize}
\item[$\bullet$]
$\sigma(i) = \pi(j)$;
\item[$\bullet$]
$\sigma(i+1) \neq \pi(j+1)$.
\end{itemize}
\label{lem:pairs}
\end{lem} 
\begin{proof}
Suppose $\pi=\sigma\phi\Lp{i_1,i_2}$. The proof follows from the simple fact that when applying a translocation $\phi {(i_1,i_2)}$ to $\sigma$, the positions $i$ described above are among 
\[
\left\{ \begin{array}{lr}
i_1 - 1, i_1, \mbox{ and } i_2,  & \mbox{ if } i_1 < i_2,\\
i_1 - 1, i_1, \mbox{ and } i_2-1, & \mbox{ if } i_1 > i_2.
\end{array} \right.
\]
\end{proof}

\begin{cor}
\label{cor:pairs}
Let $\sigma, \pi \in \Sn$ be two permutations, and assume that there exist $a \ge 0$ different positions $i$, $i \in [n-1]$, such that $\sigma(i) = \pi(j)$, but $\sigma(i+1) \neq \pi(j+1)$ for some $j \in [n-1]$. Then, $\tldist(\sigma, \pi) \geq \lceil a/3 \rceil$.  
\end{cor}

For an integer $p\ge1$, let $\mu=(1,2,\cdots, p)$ and let $\sigma_1,\sigma_2\in\mathbb{S}{(\set{p+1,\cdots,2p-1})}$. Note that 
    \begin{equation}
    \begin{split}
    \mu\circ\sigma_1 &= \Lp{1,\sigma_1(1),2, \sigma_1(2),\cdots,p-1,\sigma_1(p-1),p}\\
    \mu\circ\sigma_2 &= \Lp{1,\sigma_2(1),2, \sigma_2(2),\cdots,p-1,\sigma_2(p-1),p}.
    \end{split}
    \end{equation}

\begin{thm} \label{thm:distance} For $\mu,\sigma_1,$ and $\sigma_2$ described above,  if $\hdist\Lp{\sigma_1,\sigma_2}\ge d$, then \[ \tldist\Lp{\mu\circ\sigma_1,\mu\circ\sigma_2}\ge \left\lceil 2d/3\right\rceil.\]
    \end{thm}

\begin{proof}
Let $\pi_1=\mu\circ\sigma_1$ and $\pi_2=\mu\circ\sigma_2$. We show that the number of indices $\ell$ in $\pi_1,$ with respect to $\pi_2,$ that satisfy the conditions described in Lemma~\ref{lem:pairs} is at least $2d$. Then, the claim of the theorem follows when we apply Corollary \ref{cor:pairs} with $a=2d$.

Assume that $\sigma_1(\ell) \neq \sigma_2(\ell)$ for some $\ell \in [p-1]$. For each such $\ell$, the two indices $2 \ell - 1$ and $2 \ell$ can both serve as index $i$ in Lemma~\ref{lem:pairs}:
\begin{enumerate}
\item We have $\pi_1(2 \ell-1) = \pi_2(2 \ell-1) = \ell$, yet
\[
\sigma_1(\ell) = \pi_1(2 \ell) \neq \pi_2(2 \ell) = \sigma_2(\ell).
\]
\item
Let $j \in [p]$ be such that $\pi_1(2 \ell) = \pi_2(2 j)$. It is easy to see that $j \neq \ell$. Then, 
\[
\ell + 1 = \pi_1(2 \ell + 1) \neq   \pi_2(2 j + 1) = j + 1 \; . 
\]
\end{enumerate}
\end{proof}

Let $\mu\circ C=\set{\mu\circ \sigma : \sigma\in C} $. From Theorem \ref{thm:distance}, we have the following corollary.
\begin{cor}\label{cor:mu_code}
For integers $n$ and $p$ with $n=2p-1$, let $\mu=(1,2,\cdots, p)$ and suppose $C\subseteq \mathbb{S}{\Lp{\set{p+1,\cdots,n}}}$ is a code with minimum Hamming distance at least $\frac{3d}2$. Then $\mu\circ C$ is a code in $\Sn$ with minimum Ulam distance at least $d$ and with size $|C|$.
\end{cor}

Hence, for odd $n$, we can construct a translocation code with length $n$, minimum distance at least $d$, and size \(A_{H}\left(\frac{n-1}{2},\left\lceil \frac{3d}{2}\right\rceil\right). \) This can be easily generalized for all $n$ to get codes of size \[ A_{H}\left(\left\lceil \frac{n}{2}\right\rceil -1,\left\lceil \frac{3d}{2}\right\rceil\right). \]

By assuming that the permutation code in the Hamming metric is capacity achieving, the asymptotic rate of the constructed code becomes
 \begin{align}\label{eqRateNonRecursive}
  &\lim\frac{\ln A_{H}\left(\left\lceil \frac{n}{2}\right\rceil -1,\left\lceil \frac{3d(n)}{2}\right\rceil\right)}{\ln n!} \nonumber\\
  &\quad =\lim\frac{\ln A_{H}\left(\left\lceil \frac{n}{2}\right\rceil -1,\left\lceil \frac{3d(n)}{2}\right\rceil \right)}{\ln\left\lceil \frac{n}{2}\right\rceil !}\cdot\frac{\ln\left\lceil \frac{n}{2}\right\rceil !}{\ln n!}\nonumber\\
   &\quad =\frac{1}{2}-\frac32\lim\frac{d(n)}{n} =\frac{1}{2}-\frac32\, \delta \end{align} 
where $\delta=\lim\frac{d(n)}{n}$.  Therefore, this code construction incurs a rate loss of $(1+ \delta)/2$ when compared to the capacity, which in this case equals $1- \delta$.

The final result that we prove in order to describe a recursive interleaving procedure is related to the longest common subsequence of two sequences and the minimum Ulam distance of interleaved sequences.

\begin{lem}
\label{lem:concat}For $\sigma,\pi\in\Sn$ and $P\subseteq[n]$,
we have 
\[
\tldist(\sigma,\pi)\ge\tldist\left(\sigma_{P},\pi_{P}\right)+\tldist\left(\sigma_{Q},\pi_{Q}\right),
\]
where $Q=[n]\backslash P$.\end{lem}
\begin{proof}
Without loss of generality, assume that $\sigma$ is the identity permutation.
It is clear that $ l\Lp{\pi}\le l\Lp{\pi_{P}}+ l\Lp{\pi_{Q}}$.
Hence,
\begin{align*}
\tldist(\sigma,\pi) & =n- l\Lp{\pi}\\
 & \ge n- l\Lp{\pi_{P}}- l\Lp{\pi_{Q}}\\
 & =|P|- l\Lp{\pi_{P}}+\left|Q\right|- l\Lp{\pi_{Q}}\\
 & =\tldist\left(\sigma_{P},\pi_{P}\right)+\tldist\left(\sigma_{Q},\pi_{Q}\right).
\end{align*}
\end{proof}

\begin{lem}
For sets $P$ and $Q$ of sizes $p$ and $p-1$, respectively, let
$C'_{1}\subseteq\mathbb{S}{(P)}$ be a code with minimum Ulam distance $d$ and let $C_{1}\subset\mathbb{S}{\left(Q\right)}$ be a
code with minimum Hamming distance $3d/2$. The code $C_{1}'\circ C_{1}=\left\{ \sigma \circ \pi : \sigma \in C_{1}', \pi \in C_{1}\right\} $
has minimum Ulam distance $d$.\end{lem}
\begin{proof}
For $\sigma_{1},\sigma_{2}\in C_{1}'$ and $\pi_{1},\pi_{2}\in C_{1}$ with $(\sigma_1, \pi_1) \neq (\sigma_2, \pi_2)$, we show
that $\tldist\left(\sigma_{1}\circ \pi_{1},\sigma_{2}\circ \pi_{2}\right)\ge d$.

The case $\sigma_{1}=\sigma_{2}$ follows from a simple use of Theorem~\ref{thm:distance}. 

Assume next that $\sigma_{1}\neq \sigma_{2}$.
Then by Lemma \ref{lem:concat}, $\tldist\left(\sigma_{1}\circ \pi_{1},\sigma_{2}\circ \pi_{2}\right)\ge\tldist\left(\sigma_{1},\sigma_{2}\right)\ge d$
and this completes the proof.
\end{proof}


Let $\alpha=\frac{3}{2}$. 
For a given $n$, set $P=\left\{ 1,\cdots,\left\lceil \frac{n}{2}\right\rceil \right\} $
and set $Q=\left\{ \left\lceil \frac{n}{2}\right\rceil +1,\cdots,2\left\lceil \frac{n}{2}\right\rceil -1\right\} $.
Suppose $C_{1}'\subseteq\mathbb{S}{(P)}$ is a code with
minimum Ulam distance $d$ and $C_{1}\subseteq\mathbb{S}{\left(Q\right)}$
is a code with minimum Hamming distance $\alpha\,d$. Assuming that permutation codes
with this given minimum Hamming distance exist, we only need to provide a construction for $C_{1}'$. 
An obvious choice for $C_{1}'$ is a code with
only one codeword. Then, $C=C_{1}'\circ C_{1}$ is a code with minimum Ulam distance $d$ and cardinality
\[
A_{H}\left(\left\lceil \frac{n}{2}\right\rceil -1,\alpha\delta\right).
\]

The gap to capacity may be significantly reduced by observing that $C_{1}'$
does not have to be a code of cardinality one, and that $C_{1}'$ may be constructed recursively from shorter codes. 

To this end, let $C_{1}'=C_{2}'\circ C_{2}$ where $C_{2}'$ is
a code of length $\left\lceil \frac{n}{4}\right\rceil $ with minimum Ulam distance $d$, while $C_{2}$ is a code of length
$\left\lceil \frac{n}{4}\right\rceil -1$ and minimum Hamming distance
$\alpha d$. 

By repeating the same procedure $k$ times we obtain
a code of the form 
\begin{equation}
\Lp{\Lp{\Lp{C_{k}'\circ C_{k}} \circ C_{k-1}}\circ\cdots}\circ C_{1}, \label{eq:code-construction}
\end{equation}
where each $C_{i}$, $i \leq k$, is a code with minimum Hamming distance $\alpha d$
and length $\left\lceil \frac{n}{2^{i}}\right\rceil -1,$ and $C_{k}'$
is a code with minimum Ulam distance $d$ and length
$\left\lceil \frac{n}{2^{i}}\right\rceil $. Since each $C_{i}$ is
a permutation code in the Hamming metric with minimum distance $\alpha d$, we must have
$\left\lceil \frac{n}{2^{i}}\right\rceil -1\ge\alpha d$. To ensure
that this condition is satisfied, in \eqref{eq:code-construction}, we let $k$ be the largest value of $i$ satisfying $\frac{n}{2^{i}}-1\ge\alpha d$.
 It is easy to see that
$k=\left\lfloor \log\frac{n}{\alpha d+1}\right\rfloor $. Furthermore, we choose $C_{k}'$ to consist of a single codeword.

The asymptotic rate of the recursively constructed codes equals
\begin{align*}
&\lim\frac{1}{\ln n!}\sum_{i=1}^{k}\ln A_{H}\left(\left\lceil \frac{n}{2^{i}}\right\rceil -1,\alpha d(n) \right) \\
&\quad =\lim\sum_{i=1}^{k}\frac{\ln A_{H}\left(\left\lceil \frac{n}{2^{i}}\right\rceil -1,\alpha d(n) \right)}{\ln\left(\left\lceil \frac{n}{2^{i}}\right\rceil -1\right)!}
\frac{\ln\left(\left\lceil \frac{n}{2^{i}}\right\rceil -1\right)!}{\ln n!}\label{eq:recursive-rate}\\
 &\quad =\lim\sum_{i=1}^{k}\left(1-\frac{\alpha d(n) 2^{i}}{n}\right)2^{-i}\nonumber \\
 &\quad  =\lim\left(1-2^{-k}-\frac{\alpha d(n) k}{n}\right)\nonumber \\
 &\quad =1-2^{-\left\lfloor \log\frac{1}{\alpha\delta}\right\rfloor }-\alpha\delta \left\lfloor \log\frac{1}{\alpha\delta}\right\rfloor , \nonumber 
\end{align*}
where the last step follows from $\lim k=\left\lfloor \log\frac{1}{\alpha\delta}\right\rfloor $.
Note that this rate is roughly equal to $1-\alpha\delta \left(1+\log\frac{1}{\alpha \delta}\right)$.

\subsection{Permutation Codes in the Hamming Metric}

In the previous subsection, we demonstrated a number of constructions for translocation error-correcting codes based on
permutation codes in the Hamming metric and codes over distinct symbols. There exists a number of constructions for 
sets of permutations  with good Hamming distance, and codes with codewords containing distinct symbols. For example, in~\cite{Klove,PeterJ2010482,1291735,1228028} constructions of permutations in 
$\Sn$ using classical binary codes were presented, while other constructions rely on direct combinatorial arguments~\cite{1055142,1614092}. An example of code construction for codewords over
distinct symbols was presented in~\cite{Davydov}. There, specialized subcodes of Reed-Solomon codes were identified such that their codewords consist of distinct symbols.

In the former case, if $\code$ is a binary $[n, \lambda n, \beta n]$ code, the construction applied to $\code$ yields a subset of $\Sn$ of cardinality $2^{\lambda n}$, with minimum Hamming distance $\beta n$.  
This construction and constructions related to it may be used for permutation code design,
resulting in sets of permutations in $\Sn$ of cardinality $\exp\{\Theta(n)\}$ and 
minimum Hamming distance $\Theta(n)$. These permutations may consequently be used to 
construct permutation codes in $\Ss_{2n}$ with $\exp\{\Theta(n)\}$  codewords and minimum Ulam distance $\Theta(n)$. 

We describe a simple method for constructing sets of vectors of length $m>0$ over $[n]$
such that all entries of the vector are different, and such that the minimum Hamming distance between the vectors is large. 
In other words, we propose a novel construction for \emph{partial permutation codes} under the Hamming metric, suitable
for use in the recursive code construction described in the previous subsection.

The idea behind the proof is based on mapping suitably modified binary codewords in the Hamming space into 
partial permutations. For this purpose, let $\code$ be a binary $[N, K, D]$ code, and for simplicity of exposition, assume that  $n$ is a power of two.
Let $\bldc \in \code$. We construct a vector $\bldx = \chi(\bldc) \in ([n])^m$, where $\chi$ is a mapping as follows:
\begin{enumerate}
\item
Divide $\bldc$ into $m$ binary blocks $\bldc_1, \bldc_2, \cdots, \bldc_m$ of lengths $\log_2 n - \log_2 m$ each. Again, for simplicity,
we assume that $m$ is a power of two.
\item
For each block $\bldc_i$, $i \in [m]$, construct a vector $\bldx_i$ of length $\log_2 n$ according to the rule: 
the first $\log_2 n - \log_2 m$ bits in $\bldx_i$ equal $\bldc_i$, while the last $\log_2 m$ bits in $\bldx_i$ represent 
the binary encoding of the index $i$. Note that the integer values represented by the binary vectors $\bldx_1, \bldx_2, \cdots, \bldx_m$
are all different.
\item
Form an integer valued vector $\bldx = \chi(\bldc)$ of length $m$ over $[n]$, such that its $i$-th entry has the binary encoding specified by $\bldx_i$. 
Observe that all the integer entries of such a vector are different.  
\end{enumerate}

Now, take two vectors $\blda, \bldc \in \code$, such that their Hamming distance satisfies 
$\hdist(\blda, \bldc) \ge D$. Let $\bldx = \chi(\blda)$ and $\bldy = \chi(\bldc)$ be 
the corresponding vectors of length $m$ over $[n]$
constructed as described before. Then, there exist at least ${D}/({\log_2 n - \log_2 m})$ blocks of 
length $\log_2 n$ that are pairwise different. Therefore, the corresponding ${D}/({\log_2 n - \log_2 m})$ entries in $\bldx$ and $\bldy$ are 
pairwise different as well. 

Consider the set of vectors
\[
\cS' = \left\{ \chi(\bldc) \; : \; \bldc \in \code \right\} \; .
\]
It is straightforward to see that the set $\cS'$ has the following properties: 
\begin{enumerate}
\item
For any $\bldx \in \cS'$, all entries in $\bldx$ are different. 
\item
For any $\bldx, \bldy \in \cS'$, $\bldx \neq \bldy$, the Hamming distance satisfies $\hdist(\bldx, \bldy) \ge 
{D}/({\log_2 n - \log_2 m})$. 
\end{enumerate}

  The set $\cS'$ can be used similarly as the set $\Sn$ in the basic construction to obtain codes over $\Ss_{n+m}$ with minimum Ulam distance at least
\[
\Theta \left( \frac{D}{\log_2 n - \log_2 m} \right) \; . 
\]
Note that in this case, only $m$ numbers in the range $\{n+1, n+2, \cdots, n+m\}$ are inserted 
between the numbers in $[n]$, while the Hamming distance of the vectors is preserved. 

\begin{lem} 
\label{lemma:orders}
The parameters $N$, $n$ and $m$ are connected by the following equation: 
\[
N = m \log_2 n - m \log_2 m = m \log_2 \frac{n}{m}\; . 
\]
\end{lem}

From this lemma, if we take $m = \half n$, then $N = \half n$. By taking a code $\code$ with parameters 
$[\half n, \lambda n, \beta n]$, where $\lambda > 0$ and $\beta > 0$ are constants, we obtain a set $\cS'$ of size 
$2^{\lambda n}$ and Hamming distance $\Theta (n)$. The corresponding translocation code is able to correct $\Theta(n)$ translocation errors, and it has $2^{\lambda n}$ codewords. 


\subsection{Decoding of Interleaved Codes}

An efficient decoder implementation for the general family of interleaved codes is currently not known. For the case
of recursive codes, decoding may be accomplished with low complexity provided that the Hamming distance of the 
component permutation codes is increased from $\frac{3 d}{2}$ to $2 d$.

For simplicity of exposition, we assume $n=2p-1$ where $p$ is an integer. The case of 
even $n$ may be handled in the same manner, provided that one fixes the last symbol of all codewords.

Let $\sigma=\left(1,\hat{\sigma}(1),2,\hat{\sigma}\left(2\right),\cdots,\hat{\sigma}\left(p-1\right),p\right)\in C$
be the stored codeword and let $\pi \in\Sn$ be the retrieved
word. 

For $i\in\left[p-1\right]$, denote by $s_{i}^{\pi}$ the substring of $\pi$ that starts with element $i$ and ends with
element $i+1$. If $i+1$ appears before $i$ in $\pi$, then $s_{i}^{\pi}$ is
considered empty. For $i\in[p-1]$, let $\hat{\pi}\left(i\right)=u$
if $s_{i}^{\pi}$ contains some unique element $u$ of $\{p+1,\dotsc,n\}$. Otherwise, let $\hat{\pi}\left(i\right)=\epsilon$.

\begin{lem}
\label{lem:decoding-lemma}The permutation $\hat{\pi}$ differs from $\hat{\sigma}$ in at most $2\tldist(\sigma,\pi)$
positions.\end{lem}
\begin{proof}
Let $t=\tldist(\sigma,\pi)$. There exists a sequence $\phi_{1},\phi_{2},\cdots,\phi_{t}$
of translocations such that $\pi=\sigma\phi_{1}\phi_{2}\cdots\phi_{t}$.
For $i\in\left\{ 0,\cdots,t\right\} $, let $\pi_{i}=\sigma \phi_{1}\phi_{2}\cdots\phi_{i}$
and let $ L_{i}$ be given as 
\[
L_{i}=\left\{ j|\exists k\le i:\hat{\pi}_{k}\left(j\right)\neq\hat{\sigma}_{k}\left(j\right)\right\}.
\]
The set $ L_i$ may be viewed as the set of  elements displaced by one of the translocations
$\phi_{1},\phi_{2},\cdots,\phi_{i}$. Note that, for each
$i$, $ L_{i}\subseteq L_{i+1}.$ 

To prove the lemma, it suffices to
show that $\left |L_{t}\right|\le2t$, since $\left\{ j|\hat{\pi}\left(j\right)\neq\hat{\sigma}\left(j\right)\right\} \subseteq L_{t}$.

Let $L_{0}=\emptyset$. We show that $\left|L_{i}\right|\le\left|L_{i-1}\right|+2$
for $i\in\left[t\right]$. 

The translocation $\phi_{i}$ either moves an element of $[p]$ or an element of $\{p+1,\dotsc,n\}$. 
First, suppose that it moves an element $j$ of $[p]$. Then, $\phi_{i}$ can
affect only the substrings $s_{j-1}^{\pi_{i-1}}$ and $s_{j}^{\pi_{i-1}}$
of $\pi_{i-1}$. Next, assume that $\phi_{i}$ moves an element
of $\{p+1,\dotsc,n\}$. It can then be verified that at most two substrings of $\pi_{i-1}$ may
be affected by the given translocation. Hence, $\left|L_{i}\right|\le\left|L_{i-1}\right|+2$.
\end{proof}

Assume now that $C\subseteq\Sn$ is an interleaved code of the form
\[
C=C_{1}'\circ C_{1},
\]
where $C_{1}'=\left\{ (1,2,\cdots,p)\right\},$ and where $C_{1}$ is
a permutation code over the set $\left\{ p+1,\cdots,n\right\} $
with minimum Hamming distance $4t+1$. 

Let $\sigma \in C$ be the stored code word and $\pi \in\Sn$
be the retrieved word. Assume that $\tldist(\sigma,\pi)\le t$. The first
step of the decoding algorithm is to extract $\hat{\pi}$ from the permutation $\pi$. 
By Lemma \ref{lem:decoding-lemma}, we have
$\hdist\left(\hat{\sigma},\hat{\pi}\right)\le2t$. Since $C_{1}$ has minimum
Hamming distance $4t+1$, $\hat{\sigma}$ can be uniquely recovered from 
$\hat{\pi}$. 

Hence, for odd $d$, if $C_1$ has minimum Hamming distance $2d-1$, then $C$ has minimum Ulam distance at least $d$ and can be decoded using the described decoding scheme. The aforementioned decoding method may also be used on a recursive construction of depth larger than one by first decoding the inner-most components.

Note that decoding is accomplished through Hamming distance decoding of permutation codes, for which a number of interesting algorithms are known in literature~\cite{4401563,6033769,5219390}.

Similar to \eqref{eqRateNonRecursive}, the asymptotic rate of the code $C$ can be found to be
$\frac12-2\delta$, where $\delta=\lim\frac{d}{n}=\lim\frac{2t+1}{n}$.
For the recursive construction described in \eqref{eq:code-construction}, the asymptotic rate of the efficiently 
decodable codes outlined above equals $1-2^{-\left\lfloor \log\frac{1}{\alpha\delta}\right\rfloor }-\alpha\delta \left\lfloor \log\frac{1}{\alpha\delta}\right\rfloor $, 
with $\alpha=2$.

\textit{Remark:} Permutation codes in $\Sn$, correcting adjacent transposition errors, were thoroughly studied in~\cite{5485013}. 
We note that these codes can also be used to correct translocation errors. Indeed, every translocation can be 
viewed as a sequence of at most $n - 1$ adjacent transpositions. Therefore, any code in $\Sn$ that corrects 
$f(n)$ adjacent transpositions (for some function $f(n)$) can also correct $O(f(n)/n)$ translocations. 

It was shown in Theorem 3.1 of~\cite{5485013} that the upper bound on the rate of the code 
correcting $O(n^2)$ adjacent transpositions is zero. Such a code can also be used to correct $O(n)$ translocation errors. 
In comparison, the interleaved constructions described above can also correct $O(n)$ translocation errors, 
yet their rate is strictly larger than zero.  

The non-asymptotic and asymptotic rates of the discussed code families are compared in Figures~\ref{fig:rate_comparison1} and~\ref{fig:rate_comparison2}.

\begin{figure}\centering
\includegraphics[width=\columnwidth]{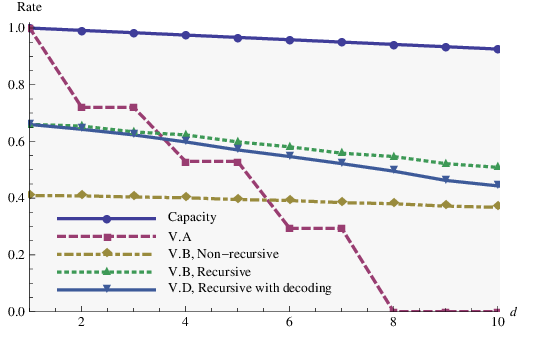}
\caption{Rate vs. distance for several code constructions with $n = 150$. The numbers in the legend refer to the section where the corresponding code is described. It is assumed that $A_H(n,d)=\frac{n!}{(d-1)!}$.}\label{fig:rate_comparison1}
\end{figure}

\begin{figure}\centering
\includegraphics[width=\columnwidth]{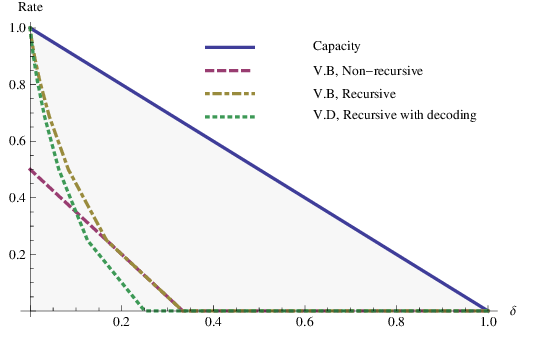}
\caption{Asymptotic rate vs. distance for several code constructions. The horizontal axis is $\delta = \lim\frac{d(n)}{n}$.}\label{fig:rate_comparison2}
\end{figure}

Note that the gap from capacity of the constructions presented in the paper is still fairly large, despite the fact that the codes are asymptotically good. This result may be attributed to the fact that the interleaving construction restricts the locations of subsets of elements in a severe manner. Alternative interleaving methods will be discussed in a companion paper. 

In what follows, we describe a method of Beame and Blais~\cite{beame2009longest} that provides translocation codes with minimum distance proportional to $n-o(n)$. This covers the zero-capacity domain of our analysis.

\subsection{Zero-rate Codes}

We present two constructions based on the longest common subsequence analysis. The first construction is based on Hadamard matrices and was given in~\cite{beame2009longest}, while the second construction is probabilistic.

Assume that a Hadamard matrix of order $k$ exists. To explain the construction, we consider permutations over the set $\left\{ 0,1,\cdots,n-1\right\} $. Furthermore, the positions in each permutation are also numbered from $0$ to $n-1$.

\begin{figure} \label{fig:hadamard}
\begin{center}
\includegraphics[width=\columnwidth]{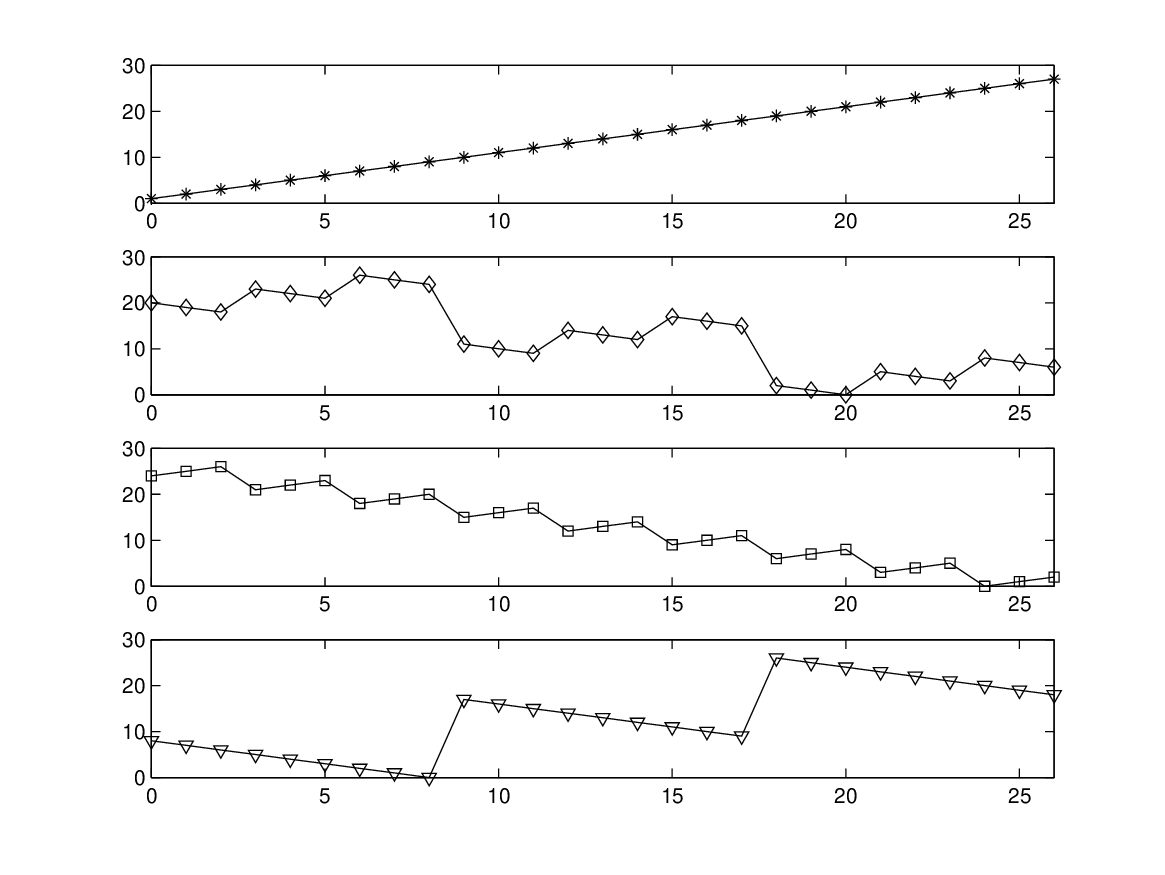}
\caption{Permutation codewords based on Hadamard matrices~\cite{beame2009longest}.}
\end{center}
\end{figure}

Let $s=\left\lceil n^{1/\left(k-1\right)}\right\rceil $. For $a\in\left\{ 0,1,\cdots,n-1\right\} $, we denote the representation of $a$ in base $s$ by $\overline{a_{1}a_{2}\cdots a_{k-1}},$ where $ $$a_{1}$ is the most significant digit. 

Let $H$ be a Hadamard matrix of order $k$ with rows and columns indexed by elements in the set $\{ 0,1,\cdots,k-1\}.$ Without loss of generality, assume the first row and column of $H$ are all-ones vectors. The set $\left\{ \pi_{i}\right\} _{i=1}^{k}$ of permutations is constructed by defining the $m$th element of $\pi_{i}$, for $m=0,1,\cdots,s^{k-1}-1$, as follows. Let $m=\overline{m_{1}\cdots m_{k-1}}$, and let the $m$th element of $\pi_{i}$ equal 
\[
\pi_{i}\left(m\right)=\overline{a_{1}a_{2}\cdots a_{k-1}},
\]
 where, for $j\in\{{0,1,\ldots,k-1\}}$, 
\[
a_{j}=\begin{cases}
m_{j}, & \qquad\mbox{if }H_{ij}=1,\\
s-1-m_{j}, & \qquad\mbox{if }H_{ij}=-1.
\end{cases}
\]
The length of the longest common subsequence of any two permutations of $\left\{ \pi_{i}\right\} $ is at most $s^{k/2-1}$. The permutations obtained in this way have length $s^{k-1}$. Consequently, the minimum distance of the code is at least $s^{k-1}-s^{k/2-1}$. Note that if $s^{k-1}>n$, we can arbitrarily delete elements from each permutation to obtain a set of permutations each of length $n$.

As an example, consider $n=27$ and $k=4$. We have 
\[
H=\left(\begin{array}{cccc}
1 & 1 & 1 & 1\\
1 & -1 & 1 & -1\\
1 & 1 & -1 & -1\\
1 & -1 & -1 & 1
\end{array}\right)
\]
 and $s=3$. Four codewords of the code based on this Hadamard matrix are plotted in Figure 6.
 
Another construction based on~\cite{beame2009longest} holds for $3\le k\le\sqrt{n}$, leading to $k$ permutations with minimum Ulam distance at least $n-32\left(kn\right)^{1/3}$. The number of codewords obtained from this construction is exponentially smaller than what may be obtained via random methods, as we demonstrate next. 

Let $U_{n}$ denote the Ulam distance between a randomly chosen permutation of length $n$ and the identity, $e=\left(1,2,\cdots,n\right)$. From a result shown by Kim~\cite{Kim1996} (see also \cite{BDJ, aldous1999longest, odlyzko2000longest}), for $0<\theta\le n^{1/3}/20$, one has
\begin{align*}
&P\left(U_{n}\le n-2\sqrt{n}-\theta n^{1/6}\right) \\
&\le\exp\left(-\theta^{3/2}\left(\frac{4}{3}-\frac{\theta}{27n^{1/3}}-\frac{5\log n}{\theta^{1/2}n^{1/3}}\right)\right).\label{eq:Un}
\end{align*}

By letting $\theta=an^{1/3}$ with $a\le1/20$, for sufficiently large
$n$, we find 
\[
P\left(U_{n}\le n-\left(2+a\right)\sqrt{n}\right)\le\exp\left(-a^{3/2}\sqrt{n}\right).
\]
Suppose a code $C$ is constructed by randomly choosing $M=e^{\alpha_{n}}$
permutations in $\mathbb{S}_{n}$, with replacement. 
By left-invariance, the bound above also holds for the Ulam 
distance between two given
codewords of $C$. Using the union bound and the fact that there are
less than $M^{2}$ pairs of codewords, the probability that there
exist two permutations with distance $\le n-(2+a)\sqrt{n}$ is bounded
from above by
\begin{align*}
M^{2}P\left(U_{n}\le n-\left(2+a\right)\sqrt{n}\right) & \le\exp\left(-a^{3/2}\sqrt{n}+2\alpha_{n}\right).
\end{align*}
To ensure that the minimum distance of the code is at least $n-\left(2+a\right)\sqrt{n}$
with high probability, we must choose $\alpha_{n}$ such that $a^{3/2}\sqrt{n}>2\alpha_{n}$.
Hence, we let $\alpha_{n}=\frac{1}{2}\sqrt{a^{3}n}-\epsilon,$ for
some $\epsilon>0$. For this choice, with high probability, the random
code $C$ of size $\Theta\left(e^{\sqrt{a^{3}n}/2}\right)$ has minimum
distance at least $n-(2+a)\sqrt{n}$. In particular, for $a=1/20$, a random code of size $\Theta\left(e^{\sqrt{n/5}/80}\right)$ with high probability has minimum distance at least $n-2.05\sqrt{n}$. 

As already pointed out, the size of a randomly constructed code obtained this way is exponential in
$\sqrt{n}$, while the size of the code from the explicit construction in~\cite{beame2009longest},
\[
\left(\frac{2+a}{32}\right)^{3}\sqrt{n},
\]
is only linear in $\sqrt{n}$.

\section{Conclusion} \label{sec:conclusion}
We introduced the notion of translocation errors in rank modulation systems. Translocation errors may be viewed as generalization of adjacent swap errors frequently encountered in flash memories. We demonstrated that 
the metric used to capture the effects of translocation errors is the Ulam distance between two permutations, a linear function of the longest common subsequence of the permutations. We also 
derived asymptotically tight upper and lower bounds on the code capacity. Furthermore, we presented a number of constructions for translocation
error-correcting codes based on interleaving permutation codes in the Hamming metric and 
deletion-correcting codes in the Levenshtein metric. Finally, we exhibited a low-complexity decoding
method for a class of relaxed interleaved codes of non-zero asymptotic rate. 

\bibliographystyle{IEEEtran}
\bibliography{bib}

\end{document}